\documentclass[final,onefignum,onetabnum]{siamart190516}
\usepackage{amssymb,graphicx,bm}
\usepackage{amsmath}
\usepackage{xcolor}
\usepackage{caption, subcaption}
\usepackage{verbatim}

\newtheorem{remark}{Remark}

\title{Optimal Control of Closed Quantum Systems via B-Splines with Carrier Waves\thanks{Submitted to the editors on \today
    \funding{This work was supported in part by LLNL Laboratory Directed Research and Development project 20-ERD-028 and in part by DOE Office of Advanced Scientific Computing Research under the Advanced Research in Quantum Computing program, award 2019-LLNL-SCW-1683.}}} 

\author{N. Anders Petersson\thanks{Center for Applied Scientific Computing, LLNL, Livermore, CA 94550, USA (\email{petersson1@llnl.gov})}
  \and Fortino Garcia\thanks{Department of Applied Mathematics, University of Colorado, Boulder, CO 80309, USA (\email{fortino.garcia@colorado.edu})}}

\begin{document}
\newcommand{\dzx}{D_0^x}
\newcommand{\dzxp}{D_{0p}^x}
\newcommand{\wdzx}{\widetilde{D_0^x}}
\newcommand{\wdzxp}{\widetilde{D_{0p}^x}}
\newcommand{\dpx}{D_+^x}
\newcommand{\dmx}{D_-^x}
\newcommand{\dzy}{D_0^y}
\newcommand{\dzyp}{D_{0p}^y}
\newcommand{\wdzy}{\widetilde{D_0^y}}
\newcommand{\dpy}{D_+^y}
\newcommand{\dmy}{D_-^y}
\newcommand{\dzz}{D_0^z}
\newcommand{\wdzz}{\widetilde{D_0^z}}
\newcommand{\dpz}{D_+^z}
\newcommand{\dmz}{D_-^z}
\newcommand{\dzt}{D_0^t}
\newcommand{\dpt}{D_+^t}
\newcommand{\dmt}{D_-^t}
\newcommand{\ehx}{E_{1/2}^x}
\newcommand{\ehy}{E_{1/2}^y}
\newcommand{\ehz}{E_{1/2}^z}
\newcommand{\calo}{{\cal O}}
\newcommand{\ab}{\bm{ a}}
\newcommand{\bb}{\bm{ b}}
\newcommand{\db}{\bm{ d}}
\newcommand{\eb}{\bm{e}}
\newcommand{\fb}{\bm{ f}}
\newcommand{\gb}{\bm{ g}}
\newcommand{\hb}{\bm{ h}}
\newcommand{\ib}{\bm{ i}}
\newcommand{\jb}{\bm{ j}}
\newcommand{\nb}{\bm{ n}}
\newcommand{\pb}{\bm{ p}}
\newcommand{\tb}{\bm{ t}}
\newcommand{\rb}{\bm{ r}}
\newcommand{\yb}{\bm{ y}}
\newcommand{\zb}{\bm{ z}}
\newcommand{\qb}{\bm{ q}}
\newcommand{\ub}{\bm{ u}}
\newcommand{\vb}{\bm{ v}}
\newcommand{\wb}{\bm{ w}}

\newcommand{\Ab}{\bm{ A}}
\newcommand{\Bb}{\bm{ B}}
\newcommand{\Eb}{\bm{ E}}
\newcommand{\Fb}{\bm{ F}}
\newcommand{\Ib}{\bm{ I}}
\newcommand{\Hb}{\bm{ H}}
\newcommand{\Kb}{\bm{ K}}
\newcommand{\Lb}{\bm{ L}}
\newcommand{\Pb}{\bm{ P}}
\newcommand{\Qb}{\bm{ Q}}
\newcommand{\Rb}{\bm{ R}}
\newcommand{\Ub}{\bm{ U}}
\newcommand{\Tb}{\bm{ T}}
\newcommand{\Xb}{\bm{ X}}

\newcommand{\Abb}{\mathbb{A}}
\newcommand{\Bbbb}{\mathbb{B}}
\newcommand{\Cbb}{\mathbb{C}}
\newcommand{\Ebb}{\mathbb{E}}
\newcommand{\Fbb}{\mathbb{F}}
\newcommand{\Ibb}{\mathbb{I}}
\newcommand{\Hbb}{\mathbb{H}}
\newcommand{\Kbb}{\mathbb{K}}
\newcommand{\Lbb}{\mathbb{L}}
\newcommand{\Pbb}{\mathbb{P}}
\newcommand{\Qbb}{\mathbb{Q}}
\newcommand{\Rbb}{\mathbb{R}}
\newcommand{\Ubb}{\mathbb{U}}
\newcommand{\Tbb}{\mathbb{T}}
\newcommand{\Xbb}{\mathbb{X}}
\newcommand{\Ybb}{\mathbb{Y}}
\newcommand{\Zbb}{\mathbb{Z}}

\newcommand{\uh}{\hat{u}}
\newcommand{\vh}{\hat{v}}
\newcommand{\ph}{\hat{p}}
\newcommand{\qh}{\hat{q}}

\newcommand{\re}{{\rm Re}\,}
\newcommand{\im}{{\rm Im}\,}

\renewcommand{\arraystretch}{1.3}
\newcommand{\p}{\partial}
\newcommand{\pdiff}[2]{\frac{\partial {#1}}{\partial {#2}}}
\newcommand{\Lag}{\mathcal{L}}

\newcommand{\eq}{\!\!\! = \!\!\!}
\newcommand{\om}{\omega}

\newcommand{\alphab}{\boldsymbol{\alpha}}
\newcommand{\lambdab}{\boldsymbol{\lambda}}
\newcommand{\psib}{\boldsymbol{\psi}}
\newcommand{\phib}{\boldsymbol{\phi}}
\newcommand{\Psib}{\boldsymbol{\Psi}}

\newcommand{\rhob}{\boldsymbol{\rho}}
\newcommand{\kab}{\boldsymbol{\kappa}}
\newcommand{\etab}{\boldsymbol{\eta}}
\newcommand{\zetab}{\boldsymbol{\zeta}}
\newcommand{\sigmab}{\boldsymbol{\sigma}}
\newcommand{\omegab}{\boldsymbol{\omega}}
\newcommand{\Gb}{\bm{ G}}
\newcommand{\kb}{\bm{ k}}
\newcommand{\sbold}{\bm{ s}}
\newcommand{\ba}{\begin{array}}
\newcommand{\ea}{\end{array}}
\newcommand{\be}{\begin{equation}}
\newcommand{\ee}{\end{equation}}
\newcommand{\bd}{\begin{displaymath}}
\newcommand{\ed}{\end{displaymath}}
\newcommand{\pa}{\partial}
\newcommand{\f}{\frac}
\newcommand{\drp}{D^r_+}
\newcommand{\drm}{D^r_-}
\newcommand{\dqp}{D^q_+}
\newcommand{\dqm}{D^q_-}
\newcommand{\dtqn}{\widetilde{{D^q_0}} }
\newcommand{\dtrn}{\widetilde{{D^r_0}} }
\newcommand{\dqn}{D^q_0}
\newcommand{\drn}{D^r_0}
\newcommand{\erh}{E^r_{1/2}}
\newcommand{\eqh}{E^q_{1/2}}

\def\dpl{D_+}
\def\dmi{D_-}

\newcommand{\ubbar}{\bar\bm{{u}}}
\newcommand{\ubar}{\bar{u}}


\newcommand{\xb}{\bm{x}}
\newcommand{\ybh}{\hat\bm{{x}}}
\newcommand{\xbh}{\hat\bm{{x}}}
\newcommand{\Ja}{J_{\alpha}}
\newcommand{\ga}{g_{\alpha}}
\newcommand{\Ma}{M_{\alpha}}

\newcommand{\Nxy}{N_{\mathbf{x}\rightarrow\hat{\mathbf{x}}}}
\newcommand{\Nyx}{N_{\hat{\mathbf{x}}\rightarrow\mathbf{x}}}
\newcommand{\Nuv}{N_{\bar{\ub}_1\rightarrow\bar{\ub}_2}}
\newcommand{\Nvu}{N_{\bar{\ub}_2\rightarrow\bar{\ub}_1}}
\newcommand{\Nij}{N_{\bar{\ub}_i\rightarrow\bar{\ub}_j}}
\newcommand{\Nji}{N_{\bar{\ub}_j\rightarrow\bar{\ub}_i}}
\newcommand{\Px}{P}
\newcommand{\Pxh}{\hat{P}}

\newcommand{\domp}{\Omega_{p}}
\newcommand{\domu}{\Omega_{\ub}}
\newcommand{\domv}{\Omega_{\vb}}
\newcommand{\domui}{\Omega_{\bar{\mathbf{u}}_i}}
\newcommand{\domuj}{\Omega_{\bar{\mathbf{u}}_j}}
\newcommand{\gamja}{\Gamma_{j0}}
\newcommand{\gamjb}{\Gamma_{j1}}
\newcommand{\gamia}{\Gamma_{i0}}
\newcommand{\gamib}{\Gamma_{i1}}

\maketitle

\begin{abstract}
We consider the optimal control problem of determining electromagnetic
pulses for implementing logical gates in a closed quantum system, where the Hamiltonian models the dynamics of coupled superconducting qudits. The quantum state is governed by Schr\"odinger's equation, which we formulate in terms of the real and imaginary parts of the state vector and solve by the St\"ormer-Verlet scheme, which is a symplectic partitioned Runge-Kutta method. A novel parameterization of the control functions based on B-splines with carrier waves is introduced. The carrier waves are used to trigger the resonant frequencies in the system Hamiltonian, and the B-spline functions specify their amplitude and phase. This approach allows the number of control parameters to be independent of, and significantly smaller than, the number of time steps for integrating Schr\"odinger's equation.

We present numerical examples of how the proposed technique can be combined with an interior point L-BFGS algorithm for realizing quantum gates, and generalize our approach to calculate risk-neutral controls that are resilient to noise in the Hamiltonian model. The proposed method is also shown to compare favorably with QuTiP/pulse\_optim and Grape-Tensorflow.

\end{abstract}



\graphicspath{{figures/}}

\section{Introduction}
We consider the quantum optimal control problem of determining electromagnetic pulses for implementing unitary gates in a quantum computer. A truncated modal expansion of Schr\"odinger's equation is used to model the quantum system, in which the state of the quantum system is described by a state vector in an $N$-dimensional Hilbert space\footnote{This paper uses conventional matrix-vector notation. Its basic relations with the bra-ket notation are $|\psi\rangle = \psib$, $\langle\psi| = \psib^\dagger$, $\langle \phi|A|\psi\rangle = \phib^\dagger A \psib$, and $| \phi \rangle \langle \psi| = \phib \psib^\dagger$.}  $\psib\in \mathbb{C}^N$. The elements of the state vector are complex probability amplitudes, where the magnitude squared of each element represents the probability that the quantum system occupies the corresponding state~\cite{Nielsen-Chuang}. Because the probabilities must sum to one, the state vector is normalized such that $\| \psib \|_2^2 = 1$. The evolution of the state vector in the time interval $t\in[0,T]$ is governed by Schr\"odinger's equation:
\begin{equation}\label{eq:schrodinger_vector}
\frac{d \psib}{d t} + {i} H(t;\boldsymbol{\alpha}) \psib = 0, \quad 0\leq t\leq T, \quad \psib(0) = \gb.
\end{equation}
Here, $i=\sqrt{-1}$ is the imaginary unit and $\gb$ is the initial state. The Hamiltonian matrix $H(t;\boldsymbol{\alpha})\in\mathbb{C}^{N\times N}$ (scaled such that Planck's constant becomes $\hbar=1$) is Hermitian and is assumed to be of the form
\begin{align}
    H(t;\alphab) = H_s + H_c(t;\alphab),
\end{align}
where $H_s$ and $H_c $ are the system and control Hamiltonians, respectively. The control Hamiltonian models the action of external control fields on the quantum system. 
The time dependence in the control Hamiltonian is parameterized by the control vector $\alphab\in\mathbb{R}^D$. As a result, the state vector $\psib$ depends implicitly on $\alphab$ through Schr\"odinger's equation.

Following Palao et al.~\cite{Palalo-2008}, we divide the Hilbert space into an essential and a guarded subspace. The essential subspace is of dimension $E$ and spans the states with the lowest energy levels in the system. These are the states used for quantum information processing. The guarded subspace has dimension $G$, such that $E+G=N$. The purpose of the guarded subspace is to eliminate any couplings to even higher states, which are excluded from the model.

The goal of the quantum optimal control problem is to determine the parameter vector $\alphab$ such that the time-dependence in the Hamiltonian matrix leads to a solution of Schr\"odinger's equation such that $\psib(T) \approx V_{tg} \gb$, where $V_{tg}$ is the target gate transformation. The gate transformation should be satisfied for all initial conditions in the essential subspace of the state vector. A basis of this subspace is provided by the matrix $U_0 \in\mathbb{R}^{N\times E}$. The definitions of $U_0$ and $V_{tg}$ are described in Appendix~\ref{app_essential}.

To account for any initial condition in the essential subspace, we define the solution operator matrix $U\in\mathbb{C}^{N\times E}$. Each column of this matrix satisfies \eqref{eq:schrodinger_vector}, leading to Schr\"odinger's equation in matrix form,
\begin{equation}\label{eq:schrodinger_matrix}
\frac{d U(t)}{d t} + {i} H(t;\boldsymbol{\alpha}) U(t) = 0, \quad 0\leq t\leq T, \quad U(0) = U_0.
\end{equation}
The overlap between the target gate matrix and the solution operator matrix at the final time is defined by $O_T:=\left\langle U(T;\alphab), V_{tg} \right\rangle_F$, where $\langle\cdot,\cdot\rangle_F$ denotes the Frobenious matrix scalar product. Because $U_0$ spans an $E$-dimensional subspace of initial conditions, we have $|O_T| \leq E$. The difference between $U(T;\alphab)$ and $V_{tg}$ can be measured by the target gate infidelity~\cite{qutip, Leung-2017, Lucarelli-2018, Machnes-2018, Shi_2019},
\begin{equation}\label{eq:objf}
  {\cal J}_1(U_T(\alphab)) := 1 - \frac{1}{E^2} \left| \left\langle U_T(\alphab), V_{tg} \right\rangle_F \right|^2,\quad U_T(\alphab) := U(T;\alphab).
\end{equation}
Note that the target gate infidelity is invariant to global phase differences between $U_T$ and $V_{tg}$. In quantum physics, the global phase of a state is considered irrelevant because it cannot be measured.

If the evolution generated by optimal control pulses should explore some of the higher states in the guarded subspace, the truncation of the Hilbert space may introduce a false nonlinearity that can misguide the optimization~\cite{Leung-2017}. Including additional states can, in principle, mitigate this problem, but is computationally expensive. Also, higher states can be difficult to model accurately. To address these challenges, we introduce an intermediate-time objective function to penalize population of states in the guarded subspace,
\begin{equation}\label{eq:objf-guard}
  {\cal J}_2(U(\cdot;\alphab)) = \frac{1}{T} \int_0^T 
  \left\langle U(t;\alphab), W U(t; \alphab) \right\rangle_F \, dt.
\end{equation}
This term is called the leakage.
Here, $W$ is an $N\times N$ Hermitian positive semi-definite matrix. To derive an appropriate expression for $W$, we start by defining a projection matrix $P_e$, such that $\psib_e = P_e \psib$ corresponds to the essential part of the state vector $\psib$. The population in the essential subspace is bounded by $0\leq \langle \psib, P_e \psib\rangle_2 \leq 1$. Maximizing the population in the essential subspace corresponds to minimizing it in the guarded subspace, defined by the orthogonal projection $P_g = (I-P_e)$. Let the guarded subspace be spanned a basis of orthonormal forbidden states~\cite{Leung-2017},  $\{\psib_{f,k}\}_{k=1}^{G}$, and construct the projection as $P_g=\sum_k \psib_{f,k} \psib_{f,k}^\dagger$. If there is more than one forbidden state, it may be desirable to introduce a state-dependent weight coefficient, leading to
\begin{align}\label{eqn::GeneralGuard}
    W=\sum_{k=1}^{G} \gamma_k\psib_{f,k} \psib_{f,k}^\dagger, \quad \gamma_k > 0.
\end{align}

From the above construction follows that if all columns of the matrix $U(t;\alphab)$ are in the essential subspace, $WU(t;\alphab)=\bm{0}$. If that equality holds for all $t\in[0,T]$, we have ${\cal J}_2 = 0$, corresponding to zero population of the guarded subspace and maximal population of the essential subspace.

For the quantum control problem with guard states, we formulate the optimization problem as
\begin{gather}
  \mbox{min}_{\bm{\alpha}}\,  {\cal G}(\alphab) := {\cal
    J}_1(U_T(\alphab)) + {\cal J}_2(U(\cdot; \alphab)), \label{eq:objf-total}\\ 
  \frac{dU}{dt} + {i} H(t; \alphab) U = 0,\quad 0\leq t\leq T, \quad
  U(0;\alphab) = U_0,\label{eq:schrodinger-matrix-essential} \\
  \alpha_{min} \leq \alpha_q \leq \alpha_{max},\quad q=1,2,\ldots,D.\label{eq_ineq-constraints2}
\end{gather}
For a discussion of the solvability of the quantum control problem, see for example Borzi et al.~\cite{Borzi-17}.

There are several different approaches to numerically determine controls that yield high fidelity quantum logic gates. Perhaps some of the simplest approaches to implement are gradient-free  optimization methods, such as the CRAB algorithm~\cite{Caneva-2011}. The CRAB algorithm is constructed by choosing a functional approximation of the control functions (Lagrange polynomials, Fourier basis, etc.) to truncate the infinite dimensional control problem into a finite dimensional, multivariate optimization problem. This optimization problem is solved using a direct search method without the need for a gradient, which typically comes at the cost of slower convergence 
compared to gradient-based techniques unless the number of control parameters is small. 

For gradient based approaches, there are broadly two classes of methods for computing the gradient of the objective function. The first accumulates the gradient by forward time-stepping, as is done in the GOAT algorithm~\cite{Machnes-2018}. This approach allows the gradient of the objective function to be calculated exactly, but requires $(D+1)$ Schr\"odinger systems to be solved when the control functions depend on $D$ parameters. This makes the method computationally expensive when the number of parameters is large. 
The second approach for computing the gradient solves an adjoint equation backwards in time. This approach, together with the use of a second order accurate Magnus scheme~\cite{HairerLubichWanner-06}, leads to the GRAPE algorithm~\cite{Khaneja-2005,Leung-2017}. Of note is that a stair-step approximation of the control functions is imposed such that each control function is constant within each time step. As a consequence, the time step determines both the numerical accuracy of the state vector, {\em and} the number of control parameters. With $Q$ control functions, and $M$ time steps of size $h$, the control functions are thus described by $D = M\cdot Q$ parameters. For problems in which the duration 
of the gate is long, or the quantum state is highly oscillatory, it is clear that the size of the parameter space becomes very large, which may hamper the convergence of the GRAPE algorithm. 

As seen with the CRAB algorithm, the total number of parameters can be reduced by instead expanding the control functions in terms of basis functions. By using the chain rule, the gradient from the GRAPE algorithm can then be used to calculate the gradient with respect to the coefficients in the basis function expansion. This approach is implemented in the GRAFS algorithm~\cite{Lucarelli-2018}, where the control functions are expanded in terms of Slepian sequences. Many other parameterizations of quantum control functions have been proposed in the literature, for example cubic splines~\cite{Ewing-1990}, Gaussian pulse cascades~\cite{Emsley-1989}, and Fourier expansions~\cite{Zax-1988}.

This paper presents a different approach, based on parameterizing the control functions by B-spline basis functions with carrier waves.
We note that the idea of using carrier waves has been used extensively for controlling molecular systems with lasers, see e.g.~\cite{ShiRab-92, CombarizaEtAl-91, Brumer-92}.
The carrier wave approach relies on the observation that transitions between the energy levels in a quantum system are triggered by resonance, at frequencies that can be determined by the system Hamiltonian. The carrier waves are used to specify the frequency spectra of the control functions, while the B-spline functions specify their envelope and phase, leading to narrow-band control functions where the frequency content is concentrated near the resonant frequencies of the system.
We find that our approach allows the number of control parameters to be independent of, and significantly smaller than, the number of time steps for integrating Schr\"odinger's equation.

The remainder of the paper is organized as follows. 
In Section~\ref{sec_resFreq}, we introduce a Hamiltonian model and 
discuss the resonant frequencies needed to trigger transitions between the states in the system. These resonant frequencies naturally motivate us to parameterize the control functions using B-splines with carrier waves; details of this parameterization are presented in Section~\ref{sec:B-Splines}.
In Section~\ref{sec_real}, we introduce a real-valued formulation of Schr\"odinger's equation and present the symplectic St\"ormer-Verlet scheme that we use for its time-integration.
To achieve an exact gradient of the discrete objective function, we apply the ``first-discretize-then-optimize'' approach. Based on the St\"ormer-Verlet scheme, in Section~\ref{sec_disc} we outline the construction of a discrete adjoint time integration method. Section~\ref{sec_numopt} presents numerical examples. We illustrate how the proposed technique, combined with the interior point L-BFGS algorithm~\cite{Nocedal-Wright} from the IPOPT package~\cite{Wachter2006}, is used to optimize control functions for multi-level qudit gates. We additionally consider a simple noise model and risk-neutral optimization to illustrate the construction of controls that are robust to uncertainty in the Hamiltonian.
The proposed scheme is implemented in the Julia~\cite{julia} programming language, in an open source package called Juqbox.jl~\cite{Juqbox-software}. In Section~\ref{sec_compare}, we compare the performance of Juqbox.jl and two implementations of the GRAPE algorithm. Concluding remarks are given in Section~\ref{sec_conc}.

\section{Hamiltonian model}\label{sec_resFreq}

The Hamiltonian for describing the quantum physics of super-conducting circuits is of the general form
\begin{align*}
    H(t) = H_s + H_c(t),
\end{align*}
where $H_s$ is the system Hamiltonian and $H_c(t)$ is the control Hamiltonian. Both terms are Hermitian matrices. The time-independent system Hamiltonian is in general not diagonal but, because it is Hermitian, always admits an eigenvalue decomposition $H_s = V \widetilde{H}_s V^\dag$, where $V$ is unitary and $\widetilde{H}_s$ is diagonal. A change of basis via the eigenvectors of $H_s$ yields the transformed Hamiltonian
\begin{align}\label{eqn::HamDiagonalized}
    \widetilde{H}(t) = \widetilde{H}_s + \widetilde{H}_c(t), \quad \widetilde{H}_c(t) = V^\dag H_c(t) V,
\end{align}
with a diagonal time-independent Hamiltonian $\widetilde{H}_s$.

The eigen-decomposition of a general $N\times N$ system Hamiltonian can be calculated numerically as long as the dimension of the Hilbert space is not too large (for a spin-chain Hamiltonian, the size limit has been estimated to $N\leq 5\cdot 10^4$ on present day computers~\cite{Sierant-2020}). Analytical formulae for diagonalizing $H_s$ are in general difficult to derive. However, approximate techniques based on Bogoliubov~\cite{Boisson-Etal-09} transformations and Schrieffer-Wolff~\cite{SW-Bravyi-11} expansions can in some cases be used to diagonalize the system Hamiltonian. For example, a quantum system with $Q$ components consisting of transmon qubits coupled to resonators can, in the dispersive limit, be modeled by (see~\cite{BlaisEtal-21} and Appendix~\ref{app_JC-Ham})
\begin{align}\label{eq_hamsys}
  H_s = \sum_{q=1}^Q \left(\omega_q a_q^\dag a_q - \frac{\xi_q}{2} a_q^{\dagger}a_q^{\dagger}a_q a_q - \sum_{p>q} \xi_{pq}  a_p^{\dagger}a_p a_q^{\dagger} a_q \right).
\end{align}
In the above formulae, $\omega_q$ is the ground state transition frequency and $\xi_q$  is the self-Kerr coefficient of subsystem $q$; the cross-Kerr coefficient between subsystems $p$ and $q$ is denoted $\xi_{pq}$. 
In the above formula, subsystem $q$ is assumed to have $n_q\geq 2$ energy levels, with lowering operator $a_q$. The lowering operator is constructed using Kronecker products,
\begin{align}
  a_q := I_{n_Q} \otimes \dots \otimes I_{n_{q+1}} \otimes  A_q \otimes I_{n_{q-1}} \otimes \dots \otimes I_{n_1} \, \in \mathbb{R}^{N\times N},\quad N= \prod_{q=1}^Q n_q,
  \end{align}
where $I_{n}$ denotes the $n\times n$ identity matrix and the single-system lowering matrix satisfies
\begin{align}
    A_q := \begin{pmatrix}
    0 & \sqrt{1}     \\
     &  \ddots   & \ddots  \\
     &           &  \ddots & \sqrt{n_q-1}  \\
     &           &         & 0   
 \end{pmatrix} \in \mathbb{R}^{n_q \times n_q}.
\end{align}
We consider a control Hamiltonian with real-valued control functions that are parameterized by the control vector $\alphab$,
\begin{align}\label{eq_hamctrl}
   H_c(t; \alphab) = \sum_{q=1}^{Q} f_q(t; \alphab) (a_q + a_q^{\dagger}),\quad
   f_q(t; \alphab) &= 2\,\mbox{Re}\left( d_q(t; \alphab)\, e^{i\omega_{r,q} t}\right).
\end{align}
where $\omega_{r,q}$ is the drive frequency in subsystem $q$.

\subsection{Rotating wave approximation}

To slow down the time scales in the state vector, we apply a rotating frame transformation 
in Schr\"odinger's equation through the unitary change of variables $\widetilde{\psib}(t) = R(t)\psib(t)$, where
\begin{align}\label{eq_rot-trans}
    R(t) = \bigotimes_{q=Q}^1\exp{\left(i\omega_{r,q} t\, A_q^\dagger A_q \right)}, 
\end{align}
and $\otimes_{q=Q}^1 C_q = C_Q \otimes C_{Q-1} \otimes \ldots \otimes C_1$. Note that we use $\omega_{r,q}$ as the frequency of rotation in subsystem $q$. The system Hamiltonian transforms into $H^{rw}_s = H_s - \sum \omega_{r,q} a_q^\dagger a_q$. Then, the rotating wave approximation is applied to transform the control Hamiltonian. Here, we substitute the laboratory frame control function $f_q(t;\alphab)$ from \eqref{eq_hamctrl} and neglect terms oscillating with frequencies $\pm 2\omega_{r,q}$. As a result, the Hamiltonians \eqref{eq_hamsys} and \eqref{eq_hamctrl} transform into (see Appendix~\ref{app_RotatingFrame} for details)
\begin{align}\label{eq_hamsysrot}
    H^{rw}_s &= \sum_{q=1}^Q \left(\Delta_q a_q^\dagger a_q -\frac{\xi_q}{2} a_q^{\dagger}a_q^{\dagger}a_q a_q - \sum_{p>q} \xi_{qp}  a_q^{\dagger} a_q a_p^{\dagger}a_p \right),\\
    H^{rw}_c(t;\alphab) &= \sum_{q=1}^Q \left(
    d_q(t; \alphab) a_q + \bar{d}_q(t; \alphab) a_q^{\dagger}\right), \label{eq_hamctrlrot}
\end{align}
where $\Delta_q = \omega_q - \omega_{r,q}$ is called the detuning frequency. The main advantages of the rotating wave approximation are the reduction of the spectral radius in the system Hamiltonian \eqref{eq_hamsysrot}, and the absence of the highly oscillatory factor $\exp(i\omega_{r,q}t)$ in the control Hamiltonian \eqref{eq_hamctrlrot}. In the following we assume that the rotating wave approximation has already been performed, and the tilde on the state vector will be suppressed. We additionally note that the target unitary $V_{tg}$ is similarly transformed into the rotating frame via $V_{tg}^{rw} = R(T)V_{tg}$.

\subsection{Resonant frequencies}\label{sec_resonant}
To simplify the presentation we restrict our analysis to a bipartite quantum system, i.e., $Q=2$. 
While our methodology applies to general Hamiltonians, the main elements of our approach can be described by a model where the system Hamiltonian is diagonal and the control Hamiltonian of the form \eqref{eq_hamctrlrot}. The general case can be handled in an analogous way by first transforming the system Hamiltonian to diagonal form, resulting in a transformed control Hamiltonian as in \eqref{eqn::HamDiagonalized}. An example of the diagonalization approach is presented in Section~\ref{sec::CNOT-Jakarta}.

In the following, we assume the system Hamiltonian to be of the form \eqref{eq_hamsysrot} and denote its elements by 
\begin{align}\label{eq_ham-diag}
    \left. H^{rw}_s \right|_{\jb, \kb} = \begin{cases}
    \kappa_{\jb},\quad & \jb = \kb,\\
    0,&\mbox{otherwise},
    \end{cases}\quad 
    \kappa_{\jb} = \sum_{q=1}^2 \left(\Delta_q j_q -\frac{\xi_q}{2}j_q(j_q-1)\right) - \xi_{12} j_1 j_2,
\end{align}
for $j_q\in [0, n_q-1]$ and where $\jb = (j_2, j_1)$ is a multi-index. 

Let us consider the case when the control functions $d_k(t)$ oscillate with carrier wave frequencies $\{\Omega_1, \Omega_2\}$, and a small amplitude $\epsilon$. These assumptions give
\begin{align}
    H^{rw}_c(t) = \epsilon H_1(t),\quad H_1(t) = 
    \sum_{k=1}^2 \left( 
    e^{i\Omega_k t} a_k + e^{-i\Omega_k t} a_k^\dagger 
    \right),\quad 0<\epsilon\ll 1.
\end{align}
We proceed by identifying the carrier wave frequencies that trigger resonance.
We make an asymptotic expansion of the solution of Schr\"odinger's equation \eqref{eq:schrodinger_vector},
$\psib = \psib^{(0)} + \epsilon \psib^{(1)} + \calo(\epsilon^2)$.
In this setting, resonance corresponds to algebraic growth in time of the amplitude of the first order perturbation.
The zero'th and first order terms satisfy
\begin{alignat}{4}
    \frac{d \psib^{(0)}}{d t} + i H^{rw}_s \psib^{(0)} &= 0,\quad && \psib^{(0)}(0) &= \gb, \label{eq_psi-zero}\\
    \frac{d \psib^{(1)}}{d t} + i H^{rw}_s \psib^{(1)} &= \fb(t),\quad && \psib^{(1)}(0) &= \bm{0}.\label{eq_psi-one}
\end{alignat}
Because the system Hamiltonian is diagonal, \eqref{eq_psi-zero} is a decoupled system of ordinary differential equation that is solved by $\psi^{(0)}_{\jb}(t) = g_{\jb} e^{-i \kappa_{\jb} t}$. The right hand side of \eqref{eq_psi-one} satisfies $\fb(t) := - i H_1(t) \psib^{(0)}(t)$, which can be written
\begin{align}
     \fb(t) = \sum_{k=1}^Q \fb^{(k)}(t),\quad
    \fb^{(k)}(t) = -i\left( 
    e^{i\Omega_k t} a_k + e^{-i\Omega_k t} a_k^\dagger 
    \right) \psib^{(0)}(t).
\end{align}

Because the matrix $H^{rw}_s$ is diagonal, the system for the first order perturbation, \eqref{eq_psi-one}, is also decoupled. We are interested in cases when $\psi^{(1)}_{\jb}(t)$ grows in time, corresponding to resonance. 
Let $\eb_k$ denote the $k^{th}$ unit vector and denote a shifted multi-index by $\jb \pm \eb_1 = (j_2, j_1 \pm 1)$ and $\jb \pm \eb_2 = (j_2 \pm 1, j_1)$.

\begin{lemma}\label{lem_freq}
The perturbation of the state vector, $\psi^{(1)}_{\jb}(t)$, grows linearly in time when the carrier wave frequencies and the initial condition satisfy:
\begin{align}
    \Omega_k &= \kappa_{\jb + \eb_k} - \kappa_{\jb},\quad 
    g_{\jb + \eb_k} \ne 0,\quad j_k \in[0,n_k-2],\label{eq_cond1}\\
\Omega_k &= \kappa_{\jb} - \kappa_{\jb - \eb_k},\quad 
g_{\jb - \eb_k} \ne 0,\quad j_k \in[1,n_k-1],\label{eq_cond2}
\end{align}
for $k=\{1,2\}$.
\end{lemma}
\begin{proof}
See Appendix~\ref{app_resonance}.
\end{proof}

We can evaluate the conditions for resonance by inserting the Hamiltonian elements from \eqref{eq_ham-diag} into \eqref{eq_cond1} and \eqref{eq_cond2}.
For $k=1$ and $j_2 \in[0,n_2-1]$, resonance occurs in $\psi^{(1)}_{\jb}(t)$ when
\begin{align}
    \Omega_1 = \begin{cases}
        \Delta_1 - \xi_1 j_1 - \xi_{12} j_2,\quad & g_{\jb+\eb_1}\ne 0,\ j_1 \in[0,n_1-2],\\
        \Delta_1 - \xi_1 (j_1-1) - \xi_{12} j_2,& g_{\jb-\eb_1}\ne 0,\ j_1 \in[1,n_1-1].
    \end{cases}
\end{align}
For $k=2$ and $j_1 \in[0,n_1-1]$, the resonant cases are
\begin{align}
    \Omega_2 = 
    \begin{cases}
        \Delta_2 - \xi_2 j_2 - \xi_{12} j_1,\quad & g_{\jb+\eb_2}\ne 0,\ j_2 \in[0,n_2-2],\\
        \Delta_2 - \xi_2 (j_2-1) - \xi_{12} j_1,\quad & g_{\jb-\eb_2}\ne 0,\ j_2 \in[1,n_2-1].
    \end{cases}
\end{align}
For example, when $n_1=3$, $n_2=3$ and $g_{\jb}\ne 0\ \forall \jb$, the carrier wave frequencies:
\begin{align*}
    \Omega_1 &= \begin{bmatrix}
    \Delta_1,& \Delta_1 - \xi_{12},& \Delta_1 - 2\xi_{12},& \Delta_1 - \xi_1,& \Delta_1 - \xi_1 - \xi_{12},& \Delta_1 - \xi_1 - 2\xi_{12}
    \end{bmatrix},\\
    \Omega_2 &= \begin{bmatrix}
    \Delta_2,& \Delta_2 - \xi_{12},& \Delta_2 - 2\xi_{12},& \Delta_2 - \xi_2,& \Delta_2 - \xi_2 - \xi_{12},& \Delta_2 - \xi_2 - 2\xi_{12}
    \end{bmatrix},
\end{align*}
lead to resonance.

\begin{remark}
The purpose of the above analysis is to identify carrier wave frequencies that trigger resonance, i.e., algebraic growth in time of the amplitude of the first order perturbation. While the asymptotic expansion is only valid in the limit of small amplitudes, it is sufficient for determining the {\em onset} of resonance. Since Schr\"odinger's equation conserves total probability, the linear growth only occurs for short times.
\end{remark}

\section{Quadratic B-splines with carrier waves}\label{sec:B-Splines}
Motivated by the results from the previous section, we parameterize the rotating frame control functions using basis functions that act as envelopes for carrier waves with fixed frequencies:
\begin{align}\label{eq_rotctrl}
  d_k(t;\alphab) &= \sum_{n=1}^{N_f} d_{k,n}(t;\alphab),\quad d_{k,n}(t;\alphab) = \sum_{b=1}^{N_b}\hat{S}_b(t) \alpha^{k}_{b,n}\,e^{it\Omega_{k,n}},\quad k\in[1,Q].
\end{align}
Here, $\Omega_{k,n}\in\mathbb{R}$ is the $n^{th}$ carrier wave frequency for system $k$. These frequencies are chosen to match the resonant frequencies in the system Hamiltonian \eqref{eq_hamsysrot}, as outlined above. 
The complex coefficients $\alpha^{k}_{b,n} = \alpha^{k (r)}_{b,n} + i \alpha^{k (i)}_{b,n}$ are control parameters that are to be determined through optimization, corresponding to a total of $D = 2 Q N_b N_f$ real-valued parameters. It is convenient to also define the real-valued functions
\begin{align}\label{eq_rotctrl_pq}
  p_{k,n}(t;\alphab) &= \sum_{b=1}^{N_b}\hat{S}_b(t) \alpha^{k (r)}_{b,n},
  \quad q_{k,n}(t;\alphab) = \sum_{b=1}^{N_b}\hat{S}_b(t) \alpha^{k (i)}_{b,n},
\end{align}
such that $d_{k,n}(t;\alphab) = \left(p_{k,n}(t;\alphab) + i q_{k,n}(t;\alphab)\right)\exp(it\Omega_{k,n})$.

The are several options for choosing the basis functions $\hat{S}_b(t)$ in \eqref{eq_rotctrl}. For example, a Fourier or Chebyshev expansion has orthogonal basis functions and the Fourier basis has compact support in frequency space. However, using either of these bases makes it difficult to construct control functions that begin and end with zero amplitude, which may be desirable in an experimental implementation. Furthermore, the global support in time implies that each coefficient in the Fourier or Chebychev expansion modifies the control function for all times, which may hamper the convergence of the optimization algorithm. An interesting alternative to these classical expansions is provided by semi-orthogonal B-spline wavelets, which can be optimally localized in both time and frequency~\cite{Unser97}. Here, each basis function is a smooth polynomial function of time, where the smoothness depends on the spline order. Polynomial functions are inexpensive to evaluate and lead to a computationally efficient implementation. The basis functions have local support in time, making it easy to construct control functions that begin and end at zero. In this work we use quadratic B-splines, where at most three B-spline coefficients influence the control function at any time $t$.

The B-spline basis functions $\hat{S}_b(t)$ are centered on a uniform grid in time 
(see Figure~\ref{fig_spline_signal}),  
\begin{equation}\label{eq_b-grid}
\tau_b = \Delta\tau(b - 1.5),\quad b=1,\ldots,N_b,\quad \Delta\tau = \frac{T}{N_b-2}.
\end{equation}
The basis function $\hat{S}_b(t)$ has local support for $t\in[\tau_b - 1.5\Delta\tau, \tau_b + 1.5\Delta\tau]$. Thus, for any fixed time $t$ a control function will get contributions from at most three B-spline
wavelets, allowing the control functions to be evaluated very efficiently. 
We also remark that the control function \eqref{eq_rotctrl} can be evaluated at any time $t\in[0,T]$. Importantly, this allows the time-integration scheme to be chosen independently of the parameterization of the control function, and allows the number of control parameters to be chosen independently of the number of time steps for integrating Schr\"odinger's equation.
\begin{figure} 
    \centering
    \includegraphics[width=0.5\linewidth]{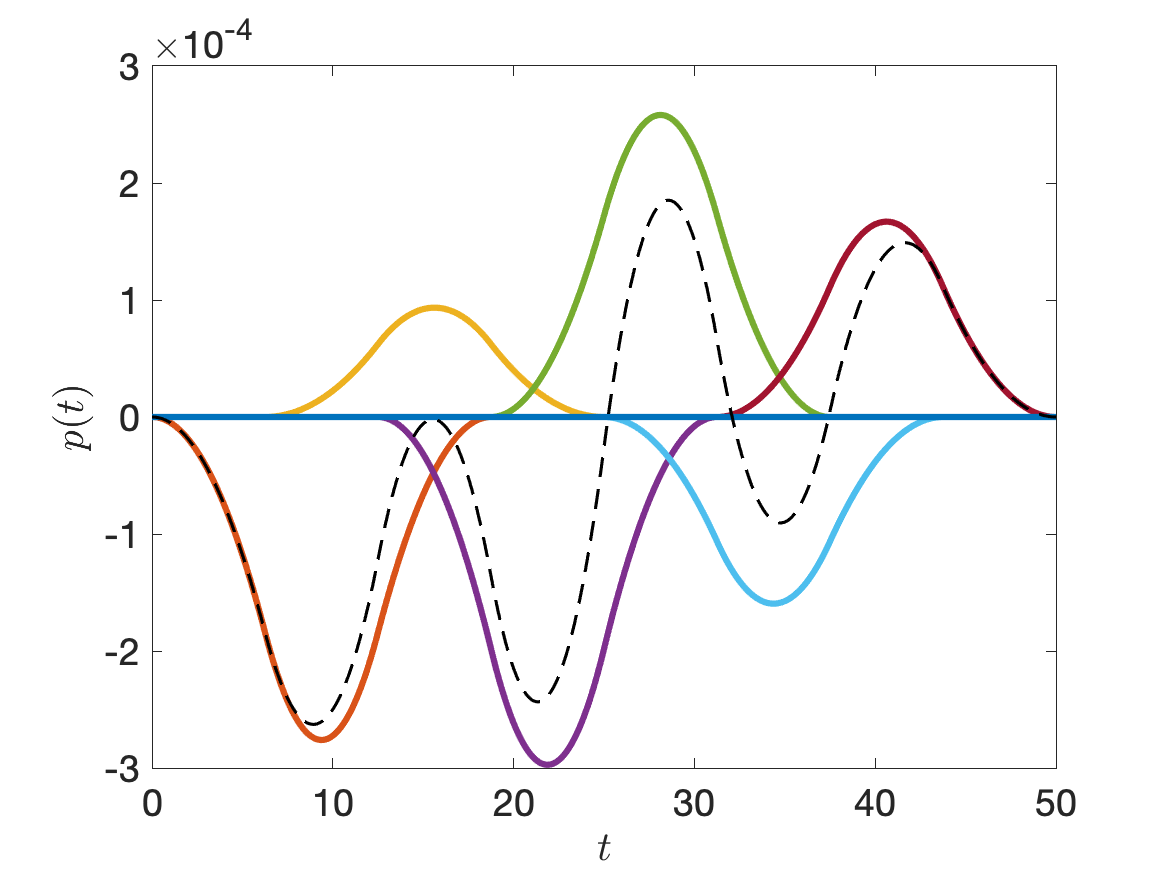}
    \caption{The real part of a quadratic B-spline control function, with zero carrier frequency (dashed black). The solid colored lines are the individual B-spline wavelets.}  \label{fig_spline_signal}
\end{figure} 

\section{Real-valued formulation}\label{sec_real}
There are many accurate numerical schemes for integrating the complex-valued Schr\"odinger equation \eqref{eq:schrodinger_vector} in time. However, by first posing it as a Hamiltonian system in terms of the real and imaginary parts of the state vector, it becomes possible to use a partitioned Runge-Kutta method~\cite{HairerLubichWanner-06}. As will be explained in more detail below, such methods can be twice as efficient for solving Schr\"odinger's equation compared to a regular (non-partitioned) Runge-Kutta scheme.

A real-valued formulation of Schr\"odinger's equation \eqref{eq:schrodinger_vector} is given by
\begin{equation}\label{eq_real-schrodinger}
  \begin{bmatrix}
    \dot{\bm{u}}\\ \dot{\bm{v}}
  \end{bmatrix} =
  \begin{bmatrix}
    S(t) & -K(t) \\ K(t) & S(t)
  \end{bmatrix}
  \begin{bmatrix} \bm{u}\\ \bm{v} \end{bmatrix} =:
  \begin{bmatrix}
    f^u(\bm{u},\bm{v},t)\\
    f^v(\bm{u},\bm{v},t)
  \end{bmatrix},\quad
  \begin{bmatrix}
    \bm{u}(0)\\
    \bm{v}(0)
  \end{bmatrix}
  =
  \begin{bmatrix}
    \bm{g}^u\\
    \bm{g}^v
  \end{bmatrix},
\end{equation}
where,
\[
  \bm{u} = \mbox{Re}(\bm{\psi}),\quad \bm{v} = -\mbox{Im}(\bm{\psi}),\quad 
  K = \mbox{Re}\,(H),\quad S = \mbox{Im}\,(H),
\]
%
Because the matrix $H$ is Hermitian, $K^T=K$ and $S^T=-S$ (note that the matrix $S$ is unrelated to
the matrix overlap function $S_V$). The real-valued formulation of Schr\"odinger's equation is a
time-dependent Hamiltonian system corresponding to the Hamiltonian functional,
\begin{equation}\label{eq_hamiltonian}
{\cal H}(\bm{u},\bm{v},t) = \bm{u}^T S(t) \bm{v} + \frac{1}{2} \bm{u}^T K(t)\bm{u} + \frac{1}{2}
\bm{v}^T K(t) \bm{v}.
\end{equation}
In general, $S(t)\ne 0$, which makes the Hamiltonian system non-separable.

In terms of the real-valued formulation, let the columns of the solution operator matrix in \eqref{eq:schrodinger_matrix} satisfy $U = \left[ \bm{u}_1 -i\bm{v}_1,\ \bm{u}_2 -i\bm{v}_2,\ \ldots,
  \bm{u}_E -i\bm{v}_E \right]$. Here, $(\bm{u}_j, \bm{v}_j)$ satisfy \eqref{eq_real-schrodinger}
subject to the initial conditions $\bm{g}^v_{\jb} = \bm{0}$ and $\bm{g}^u_{\jb} = \bm{e}_{\jb}$, where $\jb = (j_Q,j_{Q-1},\dots,j_1)$ is a multi-index
such that $j_q \in \{0,1,\dots,m_q-1\}$ and $m_q$ is
the number of essential levels of subsystem $q$.  The columns
in the target gate matrix $V_{tg}=[\db_1, \ldots, \db_E]$ correspond to
\[
V_{tg} = \left[ \bm{d}_1^u - i
  \bm{d}_1^v,\  \bm{d}_2^u - i \bm{d}_2^v,\ \ldots,  \bm{d}_E^u - i \bm{d}_E^v\right],\quad 
\bm{d}_j^u = \mbox{Re}(\bm{d}_j),\quad \bm{d}_j^v = -\mbox{Im}(\bm{d}_j).
\]
Using the real-valued notation, the two parts of the objective function \eqref{eq:objf-total} can be written
  %
  %
%
\begin{align}
{\cal J}_1(U_T(\alphab)) &=   \left(1 - \frac{1}{E^2} \left| S_V( U_T(\bm{\alpha}) ) \right| ^2
  \right),\label{eq_st-real-inf}\\
      {\cal J}_2(U(\cdot,\alphab)) &= \frac{1}{T} \sum_{j=1}^{E} \int_{0}^T \left\langle \bm{u}_j(t,\bm{\alpha}) -i \bm{v}_j(t,\bm{\alpha}), W (\bm{u}_j(t,\bm{\alpha}) - i\bm{v}_j(t,\bm{\alpha})) \right\rangle_2 \, dt,\label{eq_st-real-leak}
\end{align}
where
\begin{align}
  S_V( U_T ) &=
  \sum_{j=1}^{E} \left\langle \bm{u}_j(T,\bm{\alpha}) - i  \bm{v}_j(T,\bm{\alpha}), \bm{d}^u_j -i \bm{d}^v_j \right\rangle_2.
\end{align}

\subsection{Time integration}\label{sec_time-stepping}
Let $t_n= nh$, for $n=0,1,\ldots,M$, be a uniform grid in time where $h=T/M$ is the time step. Also let $\bm{u}^n\approx \bm{u}(t_n)$ and $\bm{v}^n\approx \bm{v}(t_n)$ denote the numerical solution on the grid.
We use the St\"ormer-Verlet (S-V) scheme to discretize the real-valued formulation of Schr\"odinger's equation. It is a partitioned Runge-Kutta (PRK) scheme that is symplectic, reversible, and second order accurate~\cite{HairerLubichWanner-06}. Starting from the initial conditions $\bm{u}^0 = \bm{g}^u$ and $\bm{v}^0 = \bm{g}^v$, the S-V scheme combines the trapezoidal and the implicit midpoint rules,
\begin{align}
    \bm{v}^{n+1/2} &= \bm{v}^n + \frac{h}{2}\left(
    K_{n+1/2} \bm{u}^n + S_{n+1/2} \bm{v}^{n+1/2}
    \right),\\
    \bm{u}^{n+1} &= \bm{u}^n + \frac{h}{2}\left( S_n \bm{u}^n + S_{n+1} \bm{u}^{n+1} - \left(K_n + K_{n+1}\right) \bm{v}^{n+1/2}\right),\\
    \bm{v}^{n+1} &= \bm{v}^{n+1/2} + \frac{h}{2}\left(K_{n+1/2}\bm{u}^{n+1} + S_{n+1/2}\bm{v}^{n+1/2}\right),
\end{align}
for $n=1,2,\ldots,M-1$.

To evaluate the computational efficiency of the S-V method, we compare it to the Implicit Midpoint Rule (IMR), which is a non-partitioned Runge-Kutta method that also is symplectic, reversible, and second order accurate. Because the matrix $S(t)\ne 0$ in \eqref{eq_real-schrodinger}, the S-V method requires two linear systems to be solved per time-step. The linear systems are diagonally dominant and can be solved using a Neumann iteration. The computational cost of each time step is dominated by solving these linear systems, which are of size $N\times N$ with band width $b$. The number of real-valued operations for evaluating the matrix-vector products is ${\cal O}(bN)$. Because linear systems must be solved to calculate $\bm{v}^{n+1/2}$ and $\bm{u}^{n+1}$ in the S-V method, the total cost of taking one time step is ${\cal O}((5 + 2 N_{iter}) b N )$, where $N_{iter}$ is the number of iterations needed to converge the solution (often $N_{iter} \leq 5$). If, instead, we would have used the non-partitioned IMR time integrator, the corresponding real-valued linear system would consist of 4 blocks, each of size $N\times N$. Each block is banded with the same band width as before. This system is also diagonally dominant and can be solved by the same iterative scheme. The cost to evaluating these matrix vector products is ${\cal O}((4 + 4 N_{iter}) b N)$. Note that the cost of solving the corresponding complex-valued system with the IMR method is also of this order because multiplying two complex numbers requires 4 real-valued multiplications. The upshot is that the S-V method is about twice as efficient as the IMR method.

\subsection{Time step restrictions for accuracy and stability}\label{sec_time-step-est}
For simplicity, in the following we consider a single qubit system ($Q=1$) and note that the analysis can be extended to multiple systems in a straightforward manner. The accuracy in the numerical solution of Schr\"odinger's equation is essentially determined by resolving each of two fundamental time scales on the grid in time. The first time scale corresponds to the resonance frequencies in the control functions for triggering desired transitions between energy levels in the quantum system (as discussed in Section~\ref{sec_resonant}). 
In the Hamiltonian model \eqref{eq_hamsysrot} and \eqref{eq_hamctrlrot}, the angular transition frequencies between the essential energy levels (with detuning frequency $\Delta_1$) are
\begin{align*}
    \Omega_{1,n} = \Delta_1 - n\xi_1,\quad n=0,\ldots,N_f-1.
\end{align*}

The second time scale is due to the harmonic oscillation of the phase in the state vector. Recall that the Hamiltonian $H \in \mathbb{C}^{N\times N}$ is Hermitian
so that if $\lambda$ is a real eigenvalue of $H$ then
the system matrix $-iH$ has the eigenvalue
$-i\lambda$. Thus, the harmonic oscillation corresponding
to the eigenvalue of largest magnitude of $H$ gives the shortest
period. A straightforward bound on the eigenvalue of largest magnitude of the Hamiltonian \eqref{eq_hamsysrot} and \eqref{eq_hamctrlrot}
can be obtained using Gershgorin's circle theorem~\cite{Golub-VanLoan},
\begin{align*}
    \rho(H)
    \leq \frac{|\xi_1|}{2}(N-1)(N-2) + 2\,d_\infty \sqrt{N-1} \equiv \rho_\text{max}.
\end{align*}
Here we have used that the control function is bounded by 
$d_\infty = \max_t |d_1(t,\alphab)|$
for a given parameter vector $\bm{\alpha}$, in the interval $0\leq t \leq T$. To resolve the shortest period in the solution of Schr\"odinger's 
equation by at least $C_P$ time steps requires
\begin{align}\label{eq_timestep}
h \leq \frac{2\pi}{C_P \max\{\rho_\text{max}, \max_{n}(|\Omega_{1,n}|)\}}.
\end{align}

The value of $C_P$ that is needed to obtain a given accuracy in the numerical solution depends on the order of accuracy, the duration of the time integration, as well as the details of the time-stepping scheme. For second order accurate methods such as the St\"ormer-Verlet method,
acceptable accuracy for engineering applications can often achieved with $C_P\approx 40$. With the St\"ormer-Verlet method, we note that the time-stepping can become unstable if $C_P\leq 2$, corresponding to a sampling rate below the Nyquist limit.

\section{Discretizing the objective function and its gradient}\label{sec_disc}
As discussed in the introduction, a powerful and efficient method to 
compute the gradient of the objective function is through the adjoint-state
method. To that end, there are two approaches to compute the 
adjoint equation. The first is the ``first-optimize-then-discretize''
approach, in which the continuous adjoint equation is derived and 
then discretized independently of the forward (state) equation. 
The second is the ``first-discretize-then-optimize'' approach, in which a discretization is chosen for both the state equation and the objective function. With the chosen discretization, the Karush-Kuhn-Tucker (KKT) conditions \cite{kuhn2014nonlinear} applied to a discrete Lagrangian yield an appropriate adjoint discretization scheme, which provides an \textit{exact} discrete gradient. Here, we follow the ``first-discretize-then-optimize" approach, as the former may yield inconsistent gradients. In the following we present a brief summary of the results of this approach. For the reader interested in the technical details of the derivation of the discrete adjoint scheme, we refer to~\cite{petersson2020discrete}.

Recall that the St\"ormer-Verlet scheme applied to the real-valued Schr\"odinger equation \eqref{eq_real-schrodinger}
is a partitioned Runge Kutta (PRK) scheme, essentially combining the trapezoidal and midpoint rules for the real and imaginary parts of the state vector, respectively. For the discretization of the guard level integral term in \eqref{eq_st-real-leak}, we use the corresponding trapezoidal and midpoint rules as follows:
\begin{align*}\label{eq:cost_guardLevel}
  {\cal J}_{2}^h(\bm{u}, \bm{v}) = \frac{h}{2T}\sum_{j=1}^{E}
    \sum_{n=0}^{M-1}
    \left(  \left \langle \bm{u}^{n}_j, W \bm{u}^{n}_j \right \rangle
    +   \left \langle \bm{u}^{n+1}_j, W \bm{u}^{n+1}_j \right \rangle
    +  2\left \langle \bm{v}^{n+1/2}_j, W \bm{v}^{n+1/2}_j \right \rangle
    \right).
\end{align*}
The above expression assumes that $W\in \mathbb{R}^{N\times N}$. The general case of a complex-valued Hermitian $W$ follows from straightforward algebra.
The superscript $h$ on ${\cal J}_2^h$ denotes the discretized objective function. We note that the trace infidelity term $\mathcal{J}_1$ in \eqref{eq_st-real-inf}, is a terminal condition with the discrete analogue $\mathcal{J}_1^h = \mathcal{J}_1$. Thus we have the discrete version of the total objective function \eqref{eq:objf-total} as
${\cal G}^h(\alphab) = {\cal J}_1^h + {\cal J}_2^h$. With these choices, a ``first-discretize-then-optimize" approach leads to a partitioned Runge-Kutta (PRK) scheme that is a consistent approximation of the continuous adjoint equation. The adjoint PRK scheme is closely related to the St\"ormer-Verlet scheme. However, the roles of the trapezoidal and midpoint rules are swapped. For example, the state variables $\ub$ are evolved with the trapezoidal rule in St\"ormer-Verlet, but the corresponding multiplier variables are evolved with the midpoint rule in the adjoint PRK scheme. In addition, the
time-dependent matrices for the adjoint scheme are evaluated at slightly different time-levels, see~\cite{petersson2020discrete} for further details.

\section{Numerical optimization}\label{sec_numopt}

Our numerical solution of the optimal control problem is based on the general purpose
optimization package IPOPT~\cite{Wachter2006}. This open-source library implements a primal-dual barrier approach for solving large-scale nonlinear programming problems, i.e., it minimizes an objective function subject to inequality (barrier) constraints on the parameter vector. Since the
Hessian of the objective function is costly to calculate, we use the L-BFGS algorithm~\cite{Nocedal-Wright} in IPOPT, which only relies on the objective function and its gradient to be evaluated. Inequality constraints that limit the amplitude of the parameter vector
$\bm{\alpha}$ are enforced internally by IPOPT.

When solving an optimal control problem, the goal is to find a control vector that simultaneously gives a small gate infidelity and a small leakage. Because we minimize their sum, it is conceivable that the optimization algorithm would find an optima where only one of the terms is small. In practice this can occur when the Hilbert space is truncated too aggressively, leading to high leakage. This problem can be circumvented by inflating the guarded subspace. It can also happen that the optimizer converges to a solution where the infidelity is large. This behavior indicates that the bounds on the control vector are set too tight for the given gate duration. Increasing the bounds, or the gate duration, have been found to be effective ways of resolving that situation.

The routines for evaluating the objective function and its gradient are implemented in the Julia programming language~\cite{julia}, which provides a convenient interface to IPOPT. 
Given a parameter vector $\bm{\alpha}$, the routine for evaluating the objective function solves the Schr\"odinger equation with the St\"ormer-Verlet scheme and evaluates the objective function ${\cal G}^h(\bm{\alpha})$ by accumulation. The routine for evaluating the gradient first applies the St\"ormer-Verlet scheme to calculate terminal conditions for the state variables. It then proceeds by accumulating the gradient $\nabla_\alpha{\cal G}^h$ by simultaneous reversed time-stepping of the discrete adjoint scheme and the St\"ormer-Verlet scheme.  These two fundamental routines, together with many support functions 
are implemented in the software package Juqbox.jl~\cite{Juqbox-software}. This package was used to generate the numerical results below.

\subsection{A CNOT gate on two qudits with guard levels}\label{sec::CNOT-Jakarta}
To test our methods on a realistic quantum optimal control problem, we consider a CNOT gate on two transmon qubits coupled by a resonator bus. Magesan and Gambetta~\cite{MagGam-20} derived an effective Hamiltonian where the resonator modes are adiabatically eliminated. This model is used to describe the superconducting systems provided by IBM Quantum Experience~\cite{IBMQ-21}. 
In the rotating frame, the Hamiltonian $H^{rw}(t) = H_s^{rw} + H_c^{rw}(t)$ satisfies
\begin{align}
    H_s^{rw} &= \sum_{q=1}^2 \left(\Delta_q a_q^\dag a_q - \frac{\xi_q}{2}a_q^\dag a_q^\dag a_q a_q \right) + J(a_1 a_2^\dag + a_1^\dag a_2),\\
    H_c^{rw}(t) &= \sum_{q=1}^2 \left(d_q(t) a_q  + \bar d_q(t) a_q^\dag\right).
\end{align}
Here, we have assumed that the frequency of rotation is the same in both subsystems, making the coupling term $(a_1 a_2^\dagger + a_1^\dagger a_2)$ invariant under the rotating frame transformation. In the following, the frequency of rotation is chosen as the average of each qubit's frequency, $\omega_r = (\omega_1 + \omega_2)/2$, resulting in the detuning frequencies $\Delta_1 = -\Delta_2 = (\omega_2-\omega_1)/2$. We use representative parameters for qubits \#3 and \#5 from the {\tt ibmq$\_$jakarta} system, with qubit frequencies $\omega_1/2\pi = 5.17839$ GHz and $\omega_2/2\pi = 5.06323$ GHz. The self-Kerr coefficients are $\xi_1/2\pi = 0.3411$ GHz and $\xi_2/2\pi = 0.3413$ GHz, and the coupling coefficient is $J/2\pi = 1.995 \cdot 10^{-3}$ GHz. Two guard levels are added to the two essential levels in each qubit, corresponding to a 16-dimensional Hilbert space for the state vector of the coupled system.

Because the system Hamiltonian is {\em not} diagonal, its resonant frequencies are determined after diagonalization, $\widetilde{H}_s = V^\dagger H_s^{rw} V$. Here we calculate the transformation numerically, but we could alternatively have utilized an approximate diagonalization, e.g., a Schrieffer-Wolff~\cite{SW-Bravyi-11, SW-original} expansion. The transition frequencies in the system follow as differences between the eigenvalues of $H_s^{rw}$, see Section~\ref{sec_resonant}. Because the coupling coefficient $J$ is small compared to the detuning frequencies, the non-diagonal part of the Hamiltonian only imposes a small perturbation on the eigenvalues. As a result, the ``dressed" detuning frequencies $\tilde{\Delta}_{1}/2\pi \approx 5.7611\cdot 10^{-2}$ GHz and $\tilde{\Delta}_{2} = - \tilde{\Delta}_{1}$ are very close to $\Delta_1/2\pi =  5.7580\cdot 10^{-2}$ GHz and $\Delta_2 = -\Delta_1$. A more significant outcome of the unitary transformation is that the control Hamiltonian in the first system becomes $V^\dagger a_1 V \approx a_1 + \epsilon \tilde{a}_2$, where $\epsilon\approx 0.01$ and $\tilde{a}_2$ has the same non-zero structure as $a_2$. Thus, in addition to triggering resonance in the first system, it can also trigger resonance in the second system (cross-resonance). We therefore apply two carrier wave frequencies in the first control function, $\Omega_{1,1} = \tilde{\Delta}_1$ and $\Omega_{1,2} = \tilde{\Delta}_2$. The transformation of the second control Hamiltonian leads to a corresponding non-zero structure and we apply the same carrier wave frequencies in that control function. Note that the purpose of the diagonalization is to determine the carrier wave frequencies that trigger resonance between the essential states in the system. Once they are determined, we can solve Schr\"odinger's equation using the original Hamiltonians.

To penalize occupation of the forbidden states that span the guarded subspace of the Hilbert space, we start by defining index ranges for the state vector. There are 4 levels in each subsystem; two essential and two guard levels. In this case, the Hilbert space is spanned by the basis states $\mathcal{I} = \{ |jk\rangle \}_{j,k = 0}^3$ (here, $|jk\rangle = \eb_j \otimes \eb_k$ where $\eb_j\in \mathbb{R}^4$ is a canonical unit vector).
The essential subspace in the coupled system is spanned by the basis states $\mathcal{E} = \{ |00\rangle, |01\rangle, |10\rangle, |11\rangle\}$ and
we span the guarded subspace by the complimentary set of states $\mathcal{F} = \mathcal{I}\setminus\mathcal{E}$. The guarded subspace is further partitioned it into  a lower guarded subspace, spanned by the states in the set $\mathcal{F}_1 = \{ |02\rangle, |12\rangle, |20\rangle, |21\rangle \}$
and a higher guarded subspace, $\mathcal{F}_2 = \mathcal{F}\setminus\mathcal{F}_1$. We discourage population of forbidden states by constructing the weight matrix $W$ as in \eqref{eqn::GeneralGuard}, using the lower forbidden states $\psi_{f,k} \in \mathcal{F}_1$ with weights $\gamma_k = 1/7000$, and the higher forbidden states $\psi_{f,k}\in \mathcal{F}_2$ with weights $\gamma_k = 1/7$.

The control functions for each carrier wave are parameterized by B-splines with $D_1=25$ basis functions resulting in a total of $D = 200$ parameters. 
The amplitudes of the control functions are limited by the constraint
\begin{equation}\label{eq_amp-const}
  \|\bm{\alpha}\|_\infty := \max_{1\leq r\leq D}|\alpha_r|\leq \alpha_{max},
\end{equation}
where we have set $\alpha_{max} = 40$ MHz and enforce the control functions to begin and end at zero. We set the gate duration to $T=250$ ns and estimate the time step using the technique in Section~\ref{sec_time-step-est}. To guarantee at least $C_P=40$ time
steps per period, we use $M=20,750$ time steps, corresponding to $h\approx 1.205\cdot 10^{-2}$ ns. As initial guess for the elements of the parameter vector $\alphab$, we use a random number generator with a uniform distribution in $[-\alpha_{max}/100, \alpha_{max}/100]$. Additionally, we set the tolerance for the overall NLP error (see Equation (5) of \cite{Wachter2006}) to $5 \cdot 10^{-2}$.

In Figure~\ref{fig_ipopt-convergence} we present
the convergence history of the optimization. We show the objective function ${\cal G}$, decomposed into ${\cal J}_{1}^h$ and ${\cal J}_{2}^h$, together with the norm of the dual infeasibility, $\|\nabla_\alpha{\cal G} - z\|_\infty$ (used by IPOPT to monitor convergence, see~\cite{Wachter2006} for details). 
\begin{figure}
  \centering
    \begin{subfigure}{0.43\textwidth}
    \begin{center}
    \includegraphics[width=1.0\linewidth]{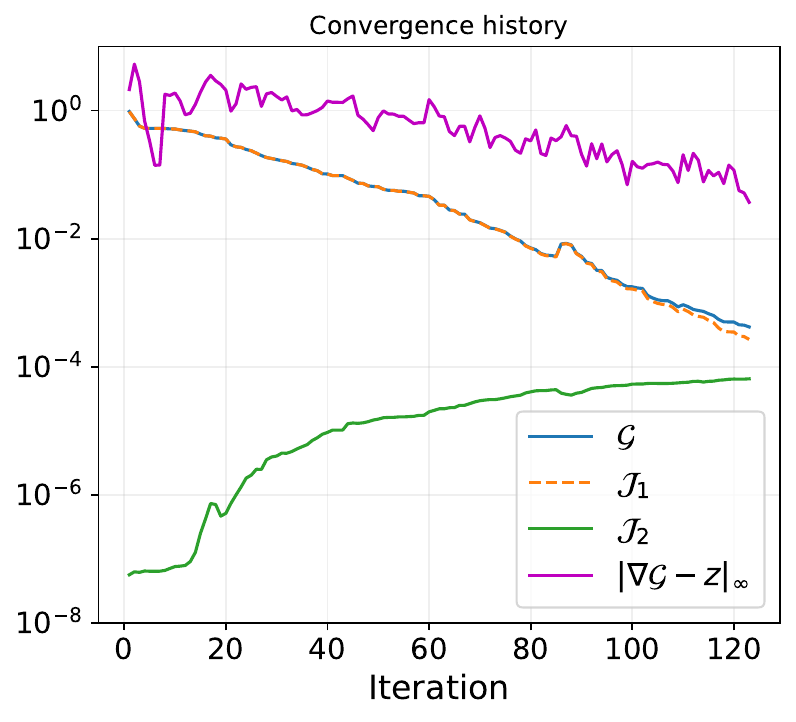}\hspace{1mm}
    \caption{IPOPT convergence history.}\label{fig_ipopt-convergence}
    \end{center}
    \end{subfigure}
    \begin{subfigure}{0.56\textwidth}
    \begin{center}
    \includegraphics[width=1.0\linewidth]{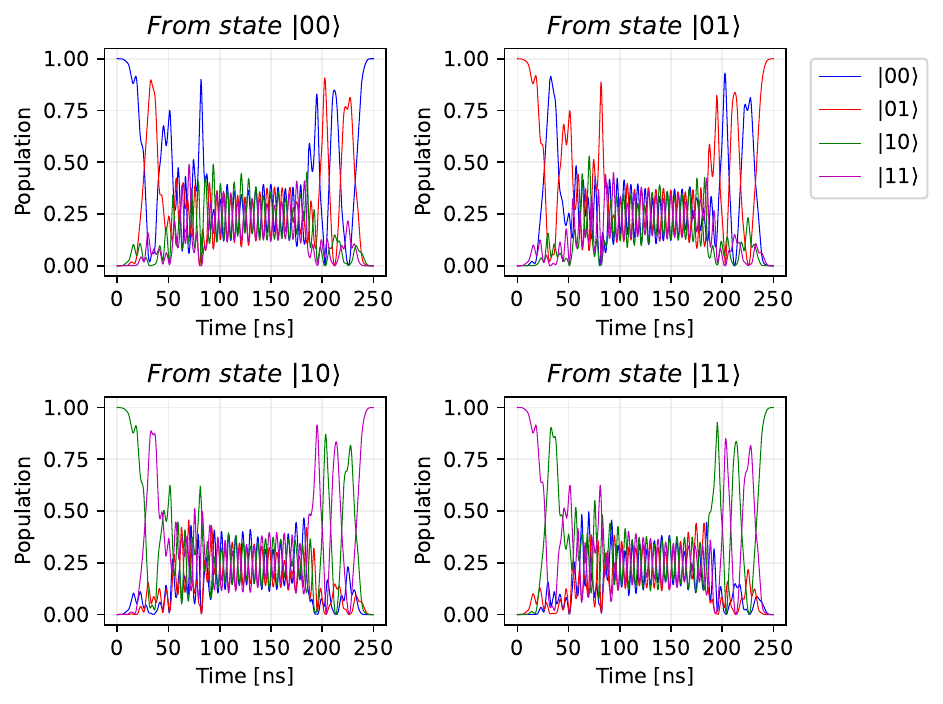}
    \caption{Population of the essential states.}\label{fig:CNOT-prob}
    \end{center}
    \end{subfigure}
  \caption{The CNOT gate on two qudits.}
\end{figure}
In this case, IPOPT converges to a solution with gate infidelity ${\cal J}_{1}^h\approx 2.68\cdot10^{-4}$ in 122 iterations. Moreover, ${\cal J}_{2}^h \approx 6.50\cdot10^{-5}$ and the population of the forbidden states remains small with a maximum population below $5.60\cdot 10^{-4}$ for all times and initial conditions. The evolution of the population of the essential states during the CNOT gate are presented in Figure~\ref{fig:CNOT-prob} and the optimized control functions are shown in Figure~\ref{fig:CNOT-ctrl}.
\begin{figure}
  \centering
  \includegraphics[width=0.75\textwidth]{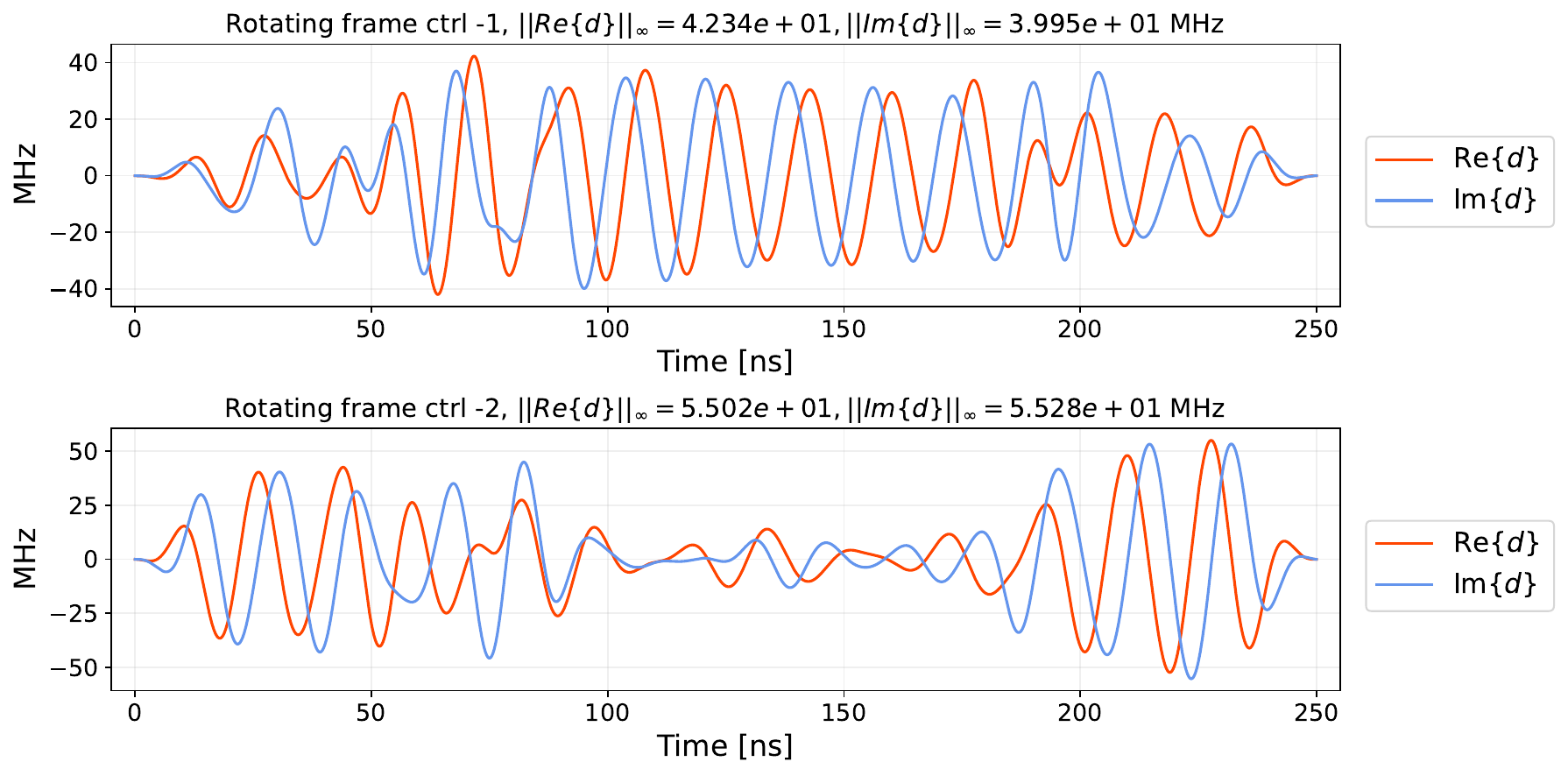}
  \caption{The control functions (in the rotating frame) for qubit $k=1$ (top), and qubit $k=2$ (bottom), for realizing a CNOT gate with a fidelity exceeding 99.9\%. The control function $d(t)$ is defined in \eqref{eq_rotctrl}. 
    }
  \label{fig:CNOT-ctrl}
\end{figure}

\subsection{Risk-neutral controls}
In practice the entries of the Hamiltonian may have some 
uncertainty, especially for higher energy levels, and it is 
desirable to design control pulses that are more robust 
to noise. There are several ways to design noise resilient controls, including robust optimization methods in which a min-max problem is solved~\cite{ge2020robust}, sampling-based learning approaches for risk-neutral optimization \cite{chen2014sampling}, or risk-neutral/averse optimization to minimize the expectation of a utility function based on the original objective function subject to uncertain parameters~\cite{Wu-robust,ge2021risk}. 
In the following we consider a risk-neutral optimization approach, which
can be interpreted as designing a control to drive an ensemble of related quantum systems to 
the same final state \cite{Li-robust,li2009ensemble}.

In this section we consider a risk-neutral strategy to design a $|0\rangle \leftrightarrow |2\rangle$ SWAP gate on a single qubit ($Q=1$), with three essential levels and one guard level. Let $\epsilon \sim \text{Unif}(- \epsilon_\text{max}, \epsilon_\text{max})$ be a uniform random variable for some $\epsilon_\text{max} > 0$. As a simple example, we consider the uncertain system Hamiltonian $H^{u}_s(\epsilon) = H^{rw}_s + H'(\epsilon)$ where $H^{rw}_s$ is given by~\eqref{eq_hamsysrot}, and $H'(\epsilon)$ is a diagonal perturbation:
\begin{align*}
    \frac{H'(\epsilon)}{2\pi} = \begin{pmatrix}
      0 &  &  &  \\
       & \epsilon/100 &  &  \\
       &  & \epsilon/10 &  \\
       &  &  & \epsilon
    \end{pmatrix}.
\end{align*}
Here, no perturbation is imposed on the control Hamiltonian \eqref{eq_hamctrlrot}. From these assumptions follow that the uncertain system Hamiltonian has expectation
$\mathbb{E}[H^{u}_s(\epsilon)] = H^{rw}_s$. We may correspondingly
update the original objective function, $\mathcal{G}(\bm{\alpha},H^{rw}_s)$, to the risk-neutral utility function $\widetilde{\mathcal{G}}(\bm{\alpha}) = \mathbb{E}[\mathcal{G}(\bm{\alpha}, H^u_s(\epsilon))]$. Given the simple form of the random variable $\epsilon$, we may compute $\widetilde{\mathcal{G}}$ by quadrature:
  \begin{align}\label{eqn::RiskNeutralObj}
    \mathbb{E}[\mathcal{G}(\bm{\alpha},H^u_s (\epsilon))] 
    = \int_{-\epsilon_\text{max}}^{ \epsilon_\text{max}}
       \mathcal{G}(\bm{\alpha}, H^u_s(\epsilon)) \, d\epsilon
       \approx
       \sum_{k = 1}^M w_k \mathcal{G}(\bm{\alpha}, H^u_s(\epsilon_k) ),
  \end{align}
where $w_k$ and $\epsilon_k$ are the weights and collocation points
of a quadrature rule. 

For the following example, we compare the optimal control obtained using the standard optimization procedure (no noise) and a risk-neutral control, in which the utility function \eqref{eqn::RiskNeutralObj} is computed using the Gauss-Legendre quadrature with $N = 9$ collocation points and 
$\epsilon_\text{max} = 10$ MHz. We set the gate duration to $T = 300$ ns, the maximum allowable amplitude to $\alpha_\text{max}/2\pi = 12$ MHz, the fundamental frequency to $\omega_1/2\pi = 4.10336$ GHz, with detuning frequency $\Delta_1=0$, and the self-Kerr coefficient to $\xi_1/2\pi = 0.2198$ GHz. 

The control functions are constructed using two carrier waves with frequencies $\Omega_{1,1} = 0$ and $\Omega_{1,2} = -\xi_1$ for both the ``noise-free" (NF) and ``risk-neutral" (RN) cases. In each case we use $D_1 = 12$ splines per control and carrier wave frequency for a total of $D = 48$ splines. We additionally constrain the controls to start and end at zero. We set the tolerance for L-BFGS to $10^{-5}$, the maximum iteration count to $150$, and use a maximum of five previous iterates to approximate the Hessian at each iteration. For the noise-free and risk-neutral optimized control functions, we use the perturbed Hamiltonian $H^u_s(\epsilon)$ to evaluate the objective function $\mathcal{G}$, for evenly spaced $\epsilon$ in the range $[-30,30]$ MHz. The results are shown in Figure~\ref{fig:risk-neutral}.
\begin{figure}
  \centering
  \includegraphics[width=0.5\textwidth]{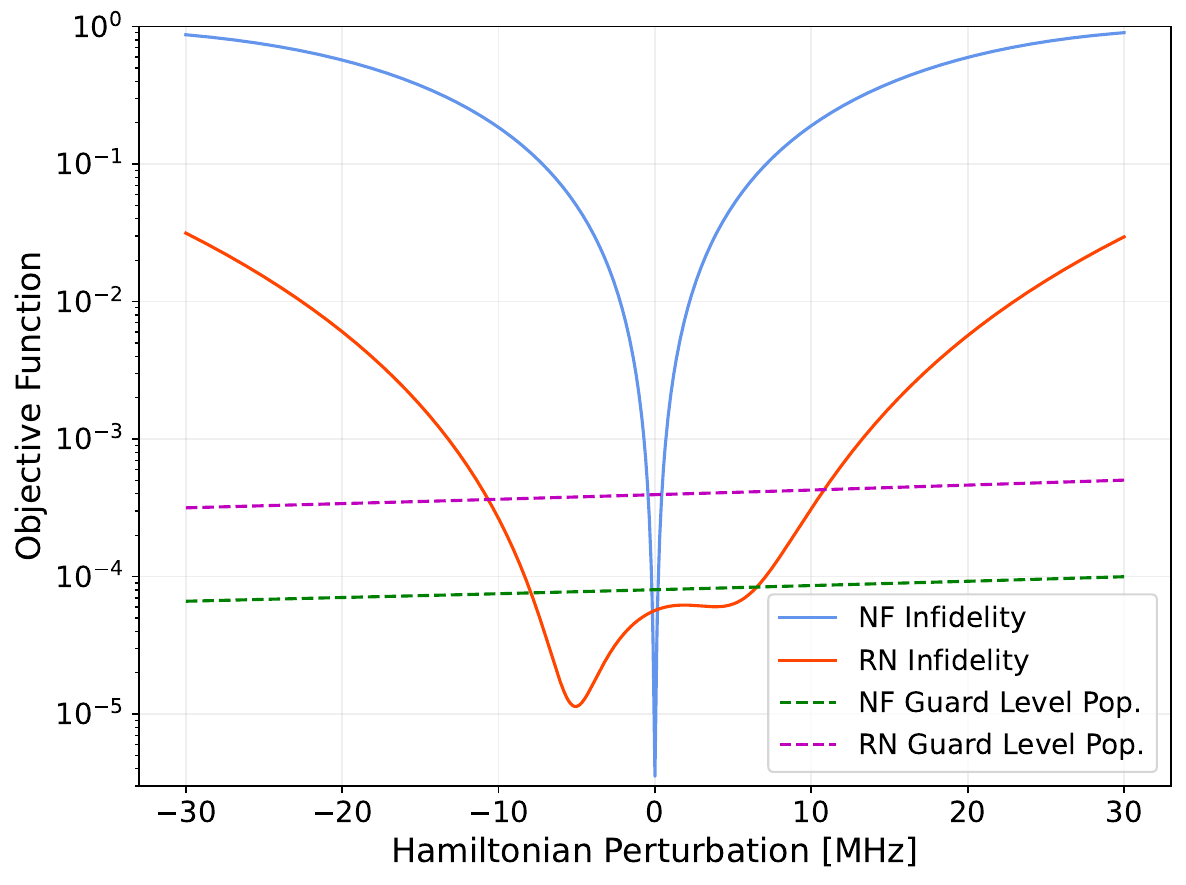}
  \caption{Infidelity objective (${\cal J}_1$) and guard level objective (${\cal J}_2$) as function of $\varepsilon$ in $H^u_s(\epsilon)$. 
  Here `NF' and `RN' correspond to the ``Noise-Free" and ``Risk-Neutral" cases.}
  \label{fig:risk-neutral}
\end{figure}
From Figure~\ref{fig:risk-neutral} we note that the optimal control corresponding to the noise-free approach obtains the smallest infidelity for $\epsilon = 0$, but it grows rapidly for $|\epsilon|>0$. By comparison, the optimal control found with the risk-neutral approach is much less sensitive to noise. 
\begin{figure}
  \centering
  \includegraphics[width=0.6\textwidth]{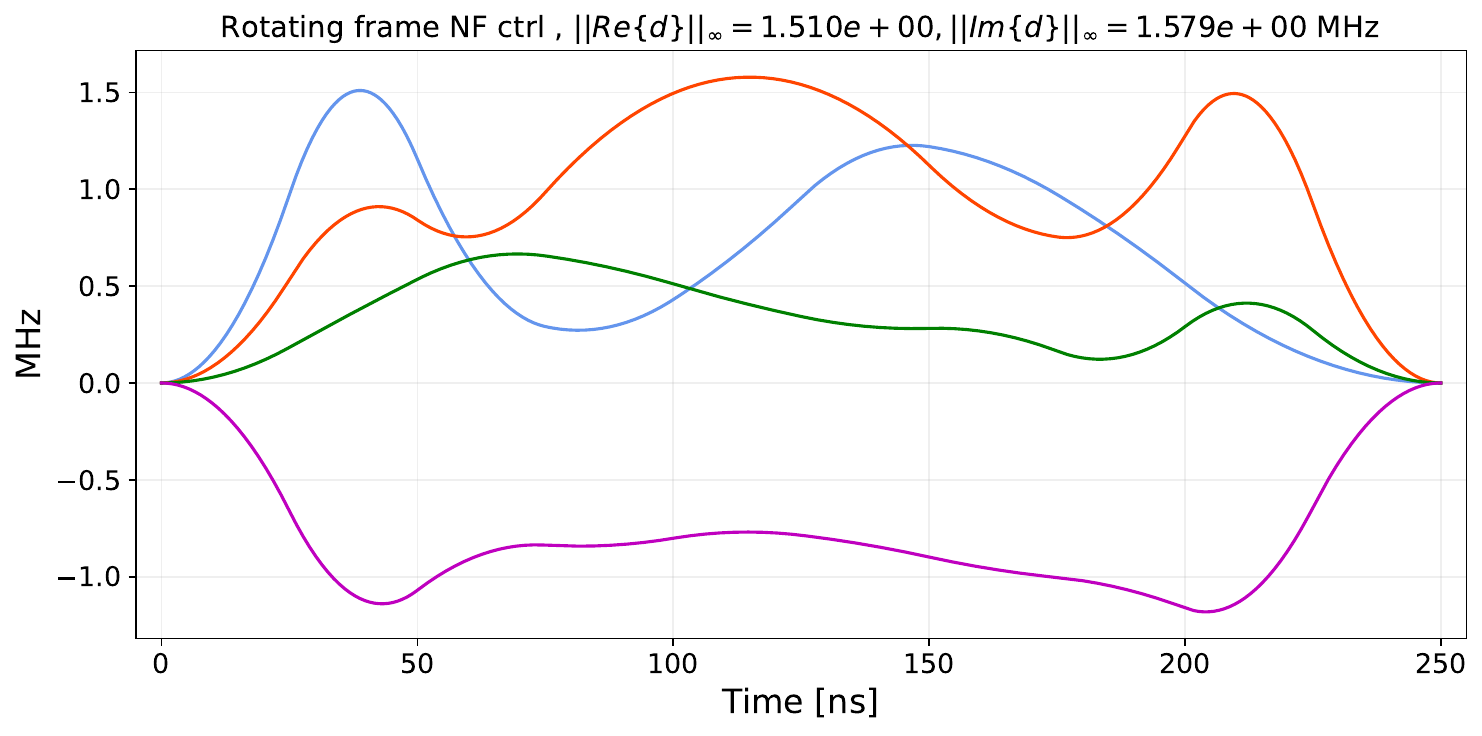}\\
  \includegraphics[width=0.6\textwidth]{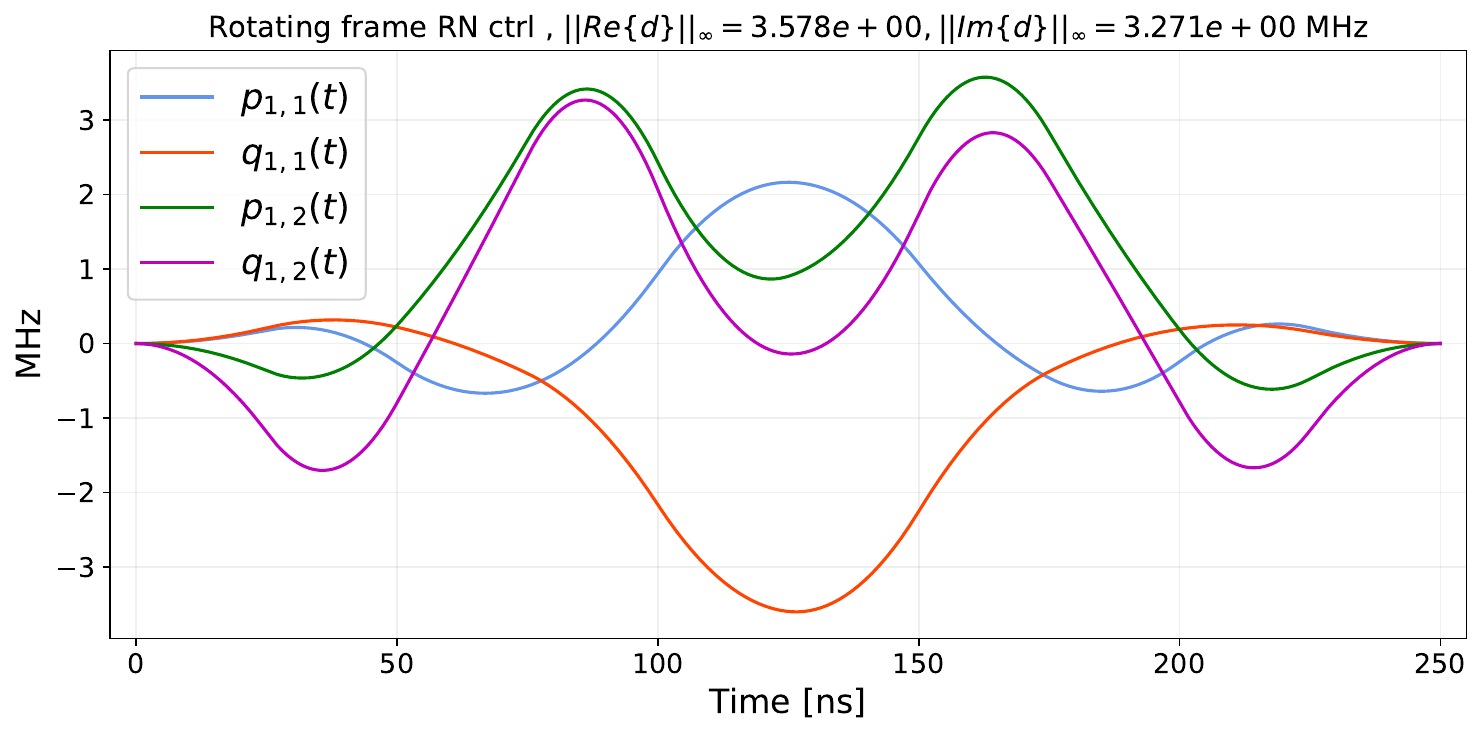}
  \caption{Control functions (without carrier waves) for the cases: ``noise-free" (top), and ``risk-neutral" (bottom). Here, $p_{k,n}(t)$ and $q_{k,n}(t)$ are defined in \eqref{eq_rotctrl_pq}.}
  \label{fig:RN-ctrls}
\end{figure}
We plot the control functions for both cases in Figure~\ref{fig:RN-ctrls}. Note that the risk-neutral controls (bottom panel) have larger amplitudes compared to the noise-free controls (top panel), indicating a potential drawback of the risk-neutral approach. However, a more systematic study of this issue is needed and left for future work.

\section{Code comparison}\label{sec_compare}

Here we compare the proposed method, as implemented in the Juqbox.jl package~\cite{Juqbox-software}, with two Python implementations of the GRAPE algorithm: the pulse\_optim method from the QuTiP \cite{qutip} framework and the Grape-TF code (TF is short for TensorFlow~\cite{abadi2016tensorflow}). The latter code provides an enhanced implementation of the GRAPE algorithm, as described by Leung et al.~\cite{Leung-2017}. It is callable from QuTiP and shares a similar problem setup with the pulse\_optim function.

As a test problem, we consider optimizing the control functions for a set of SWAP gates on a single qudit that transforms $|0\rangle \leftrightarrow |d\rangle$. The corresponding target gate transformation $V_E$ equals the permutation matrix that swaps columns $0$ and $d$ in a $(d+1)\times(d+1)$ matrix.
To evaluate the amount of leakage to higher energy
levels, we add one guard level ($G=1$) and evolve a total of $N=d+2$ states in Schr\"odinger's equation. Note that the guard level is left unspecified in the target gate transformation. We implement the SWAP gates on a multi-level qudit with transition frequency $\omega_1/2\pi = 4.8$ GHz and self-Kerr coefficient $\xi_1/2\pi = 0.22$ GHz. We apply the rotating wave approximation, with detuning $\Delta_1=0$, and model the dynamics with the system Hamiltonian \eqref{eq_hamsysrot}. As a realistic model for current superconducting quantum devices, we impose the control amplitude restrictions
\begin{align} \label{eq_amp-bounds}
  \mbox{max}_t|d(t;\alphab)| \leq c_\infty,\quad
  \frac{c_\infty}{2\pi} = 9\,\mbox{MHz},
\end{align}
in the rotating frame of reference.

\subsection{Setup of simulation codes}
QuTiP/pulse\_optim can minimize the target gate fidelity, ${\cal J}_1$, but does not suppress occupation of higher energy states. Thus, it does {\em not} minimize terms of the type ${\cal J}_2$.  As a proxy for ${\cal J}_2$, we append one additional energy level to the simulation and measure its occupation as an estimate of leakage to higher energy states.  In pulse\_optim, the control functions are discretized on the same grid in time as Schr\"odinger's equation and no smoothness conditions are imposed on the control functions. 

Grape-TF also discretizes the control functions on the same grid in time as Schr\"odinger's equation. It minimizes an objective function that consists of a number of user-configurable parts. In our test, we minimize the gate infidelity (${\cal J}_1$) and the occupation of one guard (forbidden) energy
level (similar to ${\cal J}_2$). To smooth the control functions in time, the objective function also includes terms to minimize their first and second time derivatives. 
The gradient of the objective function is calculated using the automatic differentiation (AD) technique, as implemented in the TensorFlow package. 

In Juqbox, we trigger the first $d$ transition frequencies in the system Hamiltonian by using $d$ carrier waves in the control functions, with frequencies
\begin{align*}
    \Omega_{1,k} = (k-1) (-\xi_1),\quad k=1,2,\ldots,N_f,\quad N_f = d.
\end{align*}

For all three codes, a pseudo-random number generator is used to construct the initial guesses for the parameter vector.

The pulse\_optim and Juqbox simulations were run on a Macbook Pro with a 2.6 GHz Intel iCore-7 processor. To utilize the GPU acceleration in TensorFlow, the Grape-TF simulations were run on one node of the Pascal machine at Livermore Computing, where each node has an Intel XEON E5-2695 v4 processor with two NVIDIA P-100 GPUs.

\subsection{Numerical results}
A SWAP gate where the control functions meet the control amplitude bounds \eqref{eq_amp-bounds} can only be realized if the gate duration is sufficiently long. Furthermore, the minimum gate duration increases with $d$. For each value of $d$, we used numerical experiments to determine a duration $T_d$ such that at least two of the three simulation codes could find a solution with a small gate infidelity. For Juqbox, we estimated the number of time steps using the technique in Section~\ref{sec_time-step-est} with $C_P=80$. The number of control parameters follow from $D=2N_f D_1$, where
$N_f=d$ is the number of carrier wave frequencies used and $D_1$ is the number of B-splines per control functions. We use $D_1=10$ for $d=3,4,5$ and $D_1=20$ for $d=6$. For pulse\_optim and Grape-TF, we calculate the number of time steps based on the shortest transition period, corresponding to the highest transition frequency in the system. We then use 40 time steps per
shortest transition period to resolve the control functions. For both GRAPE methods there are 2 control parameters per time step. The main simulation parameters are given in Table~\ref{tab_params}.
\begin{table}[tph]
    \begin{center}
\begin{tabular}{r|r||r|r||r|r}
\multicolumn{2}{c||}{} & \multicolumn{2}{c||}{\# time steps} & \multicolumn{2}{c}{\# parameters} \\ \hline
$d$ & $T_d$ [ns] & Juqbox & GRAPE & Juqbox & GRAPE\\ \hline
3   & 140        & 14,787  & 4,480 & 60 & 8,960 \\ \hline 
4   & 215        & 37,843  & 7,568 & 80  & 15,136 \\ \hline 
5   & 265        & 69,962  & 11,661 & 100 & 23,322 \\ \hline 
6   & 425        & 157,082 & 22,441 & 240 & 44,882 \\ \hline 
\end{tabular}
\caption{Gate duration, number of time steps ($M$) and total number of control parameters ($D$) in the $|0\rangle \leftrightarrow |d\rangle$ SWAP gate simulations. The number of time steps and control parameters are the same for pulse\_optim and Grape-TensorFlow.}\label{tab_params}
\end{center}
\end{table}

Optimization results for the pulse\_optim, Grape-TF and Juqbox codes are presented in Figure~\ref{fig_plot+table}.
The pulse\_optim code generates piecewise constant control functions that are very noisy and may therefore be hard to realize experimentally. To obtain a realistic estimate of the resulting dynamics, we interpolate the optimized control functions on a grid with 20 times smaller time step
and use the \verb|mesolve()| function in QuTiP to calculate the evolution of the system from each initial state. We then evaluate the gate infidelity using the states at the final time.
Since the control functions from Grape-TF and
Juqbox are significantly smoother, we report the gate fidelities as calculated by those codes.
\begin{figure}[htp]
        \begin{center}
    \begin{subfigure}{0.49\textwidth}
        \includegraphics[width=1.0\linewidth]{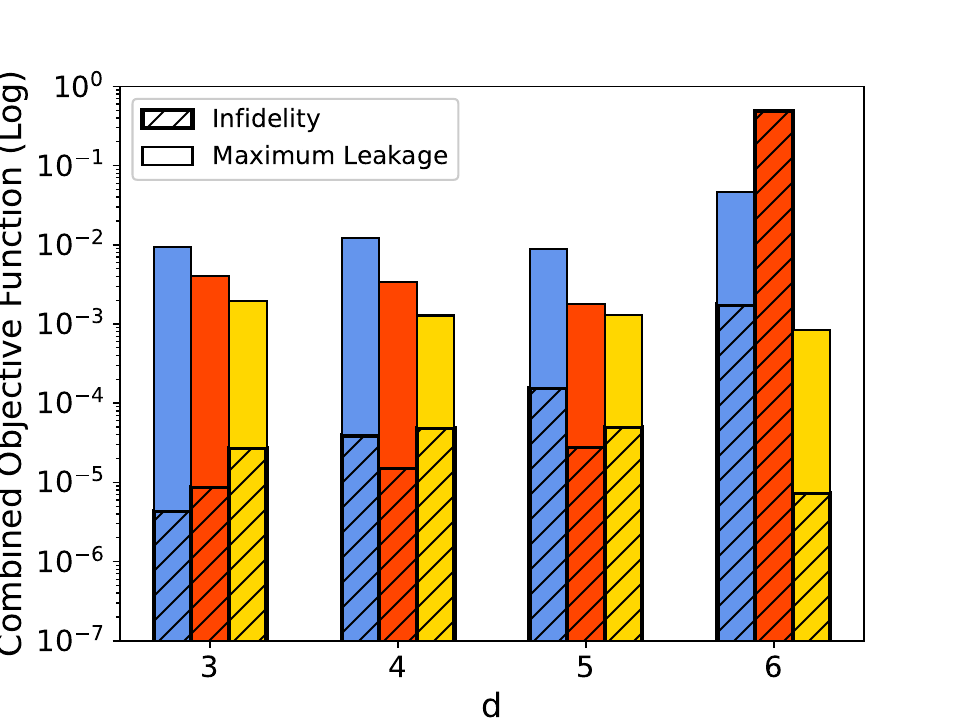}
        \caption{Infidelity and maximum leakage.}
    \end{subfigure}
    \begin{subfigure}{0.49\textwidth}
        \begin{center}
        \includegraphics[width=1.0\linewidth]{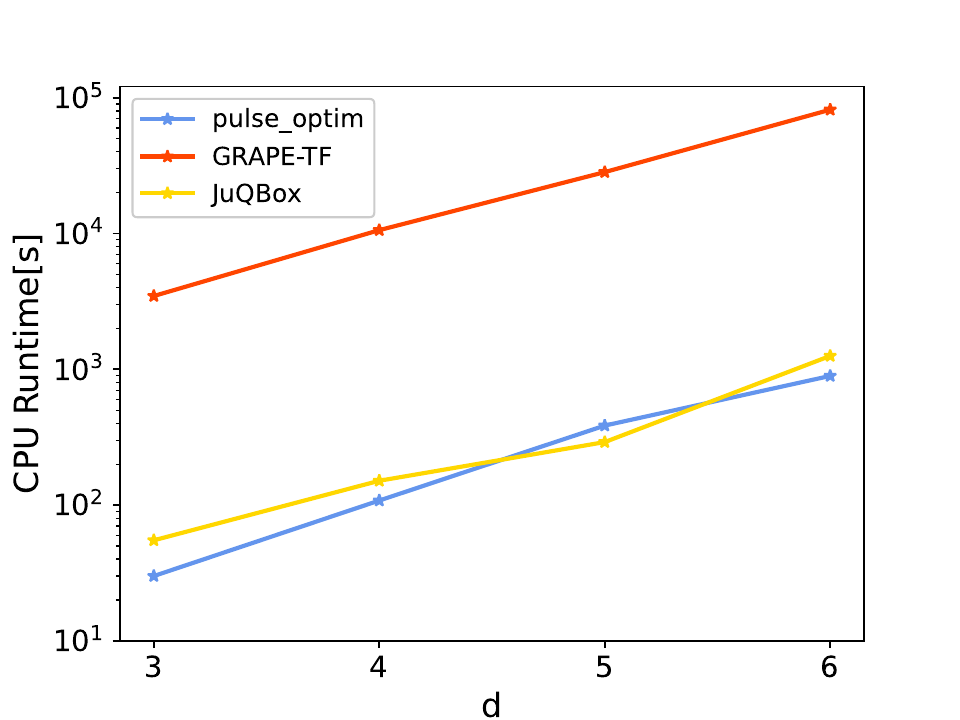}
    \end{center}
    \caption{CPU times [sec].} 
  \end{subfigure}
  \caption{Optimal control results for SWAP $|0\rangle \leftrightarrow |d\rangle$ gates using pulse\_optim (blue), Grape-TensorFlow (red) and Juqbox (yellow). (a) Gate infidelity and maximum occupation of the forbidden state $|d+1\rangle$; (b) CPU timings.}\label{fig_plot+table}
    \end{center}
\end{figure}

For the $|0\rangle \leftrightarrow |3\rangle$, $|0\rangle \leftrightarrow |4\rangle$ and $|0\rangle \leftrightarrow |5\rangle$ SWAP gates, all three codes produce control functions with very small gate infidelities. 
The most significant difference between the codes occurs for the
$d=6$ SWAP gate. Here, the Grape-TF code fails to produce a small gate infidelity after running for almost 23 hours and the pulse\_optim code results in a gate fidelity that is about 2 orders of magnitude larger than Juqbox.
While pulse\_optim and Juqbox require comparable amounts of CPU time to converge, the Grape-TF code is between 50-100 times slower despite the GPU acceleration.

To compare the smoothness of the optimized control functions, we study the Fourier spectra of the laboratory frame control functions, see Figure~\ref{fig_ctrl-fft}. Note that pulse\_optim produces a significantly noisier control function compared to the other two codes. The control function from Grape-TF is significantly smoother, even though its spectrum includes some noticeable peaks at frequencies that do not correspond to transition frequencies in the system. The Juqbox simulation results in a laboratory frame control function where each peak in the spectrum corresponds to a transition frequency in the Hamiltonian.

\begin{figure}[htp]
\begin{center}
    \begin{subfigure}{0.49\textwidth}
    \begin{center}
    \includegraphics[width=1.0\linewidth]{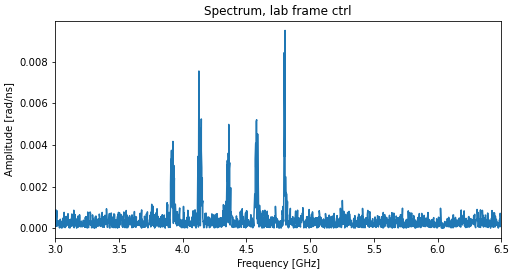}
    \caption{QuTiP/pulse\_optim.}
    \end{center}
    \end{subfigure}
    \begin{subfigure}{0.49\textwidth}
    \begin{center}
    \includegraphics[width=1.0\linewidth]{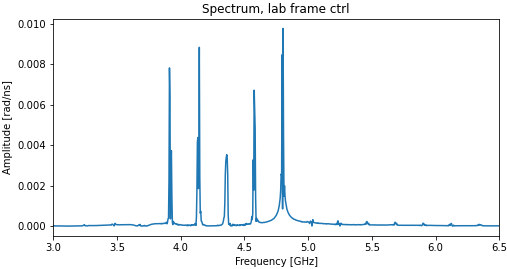}
    \caption{Grape-Tensorflow.}
    \end{center}
    \end{subfigure}
    \begin{subfigure}{0.49\textwidth}
    \begin{center}
    \includegraphics[width=1.0\linewidth]{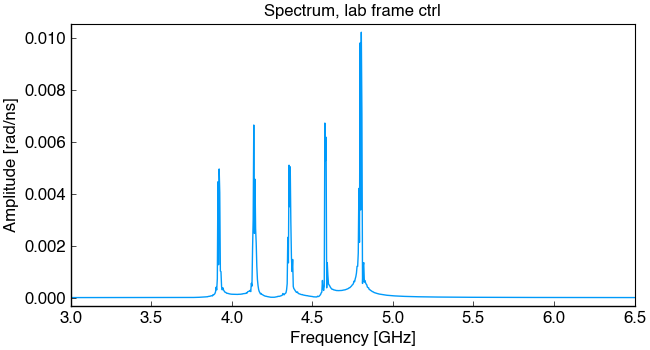}
    \caption{Juqbox.}
    \end{center}
    \end{subfigure}
\caption{Magnitude of the Fourier spectrum of the laboratory frame control function for the $|0\rangle \leftrightarrow |5\rangle$ SWAP gate on a single qudit.}\label{fig_ctrl-fft}    
\end{center}
\end{figure}

\section{Conclusions}\label{sec_conc}

In this paper we developed numerical methods for optimizing control functions for realizing logical gates in closed quantum systems where the state is governed by Schr\"odinger's equation. By asymptotic expansion, we calculated the resonant frequencies in the system Hamiltonian, corresponding to transitions between energy levels in the state vector. We introduced a novel parameterization of the control functions using B-spline basis functions that act as envelopes for carrier waves, with frequencies that match the transition frequencies. This approach allows the number of control parameters to be independent of, and significantly smaller than, the number of time steps for integrating Schr\"odinger's equation. 

The objective function in the optimal control problem consists of two parts: the infidelity of the final gate transformation and a time-integral for evaluating leakage to higher energy levels. We apply a ``first-discretize-then-optimize" approach and outline the derivation of the discrete adjoint equation that is solved to efficiently calculate the gradient of the objective function.

To demonstrate our approach, we optimized the control functions for a CNOT gate on two coupled qudits, resulting in a gate fidelity exceeding 99.9\%. Based on a simple noise model, we also generalized the proposed method to calculate risk-neutral controls that are resilient to uncertainties in the Hamiltonian model. The results are promising and indicate that a more systematic study of optimization under uncertainty can yield controls that are robust to noise in quantum systems. We finally compared the performance of the proposed method, implemented in the Juqbox package~\cite{Juqbox-software}, and two implementations of the GRAPE algorithm: the pulse\_optim method in QuTiP~\cite{qutip} and Grape-TensorFlow~\cite{Leung-2017}. The codes were compared on a set of SWAP gates on a single qudit. Here, Juqbox was found to run 50-100 times faster than Grape-TensorFlow and produce control functions that are significantly smoother than pulse\_optim. The most significant difference between the codes occurred for the $|0\rangle \leftrightarrow |6\rangle$ SWAP gate, where the Grape-TF code failed to solve the optimal control problem.

In future work, we intend to generalize our approach to solve optimal control problems for larger quantum systems.

\section*{Acknowledgment} 
We would like to thank Prof.~Daniel Appel\"o for bringing the St\"ormer-Verlet method to our attention. We also thank Dr.~Bj\"orn Sj\"ogreen for fruitful discussions on resonant frequencies, as well as two anonymous referees who helped improve the presentation.

This work was performed under the auspices of the U.S. Department of Energy by Lawrence Livermore National Laboratory under Contract DE-AC52-07NA27344. This is contribution LLNL-JRNL-823853. The views and opinions of the authors do not necessarily reflect those of the U.S. government or Lawrence Livermore National Security, LLC neither of whom nor any of their employees make any endorsements, express or implied warranties or representations or assume any legal liability or responsibility for the accuracy, completeness, or usefulness of the information contained herein.

\appendix

\section{Composite quantum systems and essential states}\label{app_essential}

To simplify the notation we assume a bipartite quantum system ($Q=2$); the case $Q=1$ is trivial and $Q>2$ follows by straightforward generalizations. Let the number of energy levels in the subsystems be $n_1$ and $n_2$, respectively, for a total of $N=n_1\cdot n_2$ states in the coupled system. We use the canonical unit vectors $\eb^{(n_q)}_j\in\mathbb{R}^{n_q}$, for $j=0,\ldots, n_q-1$, as a basis for subsystem $q$, where the superscript indicates its size. These basis vectors can be used to describe the total state of the coupled system,
\begin{align}
  \psib = \sum_{j_2 = 0}^{n_2-1} \sum_{j_1 = 0}^{n_1 -1}
  \psi_{j_2, j_1} \left(\eb_{j_2}^{(n_2)} \otimes
\eb_{j_1}^{(n_1)}\right) = \sum_{k=0}^{N-1} \vec{\psi}_{k} \eb_k^{(N)}.
\end{align}
Here, $\vec{\psi}\in\mathbb{C}^N$ denotes the one-dimensional representation of the two-dimensional state vector $\psib$, using a natural ordering of the elements, i.e., $\vec{\psi}_k = \psi_{j_2, j_1}$ for $k = j_1 + n_1 j_{2} =:k_{ind}(j_2, j_1)$. The mapping $k = k_{ind}(j_2,j_1)$ is invertible for $k\in[0,N-1]$.

We classify the energy levels in the total state vector as either essential or guarded levels. The unitary gate transformation is only specified for the essential levels. The guard levels are retained to justify the truncation of the modal expansion of Schr\"odinger's equation, and to avoid leakage of probability to even higher energy levels.

Let the number of essential energy levels in the subsystems be $m_1$ and $m_2$, respectively, where $0<m_q \leq n_q$. Similar to the total state vector, we use the canonical unit vectors as a basis for the essential subspace of each subsystem. The total number of essential levels equals $E = m_1\cdot m_2$. Let the essential energy levels in the total state vector be represented by the essential state vector $\phib$. Similar to the full state vector, we flatten its two-dimensional indexing using a natural ordering,
\begin{align}
  \phib = \sum_{i_2 = 0}^{m_2-1} \sum_{i_1 = 0}^{m_1 -1}
  \phi_{i_2, i_1} \left(\eb_{i_2}^{(m_2)} \otimes \eb_{i_1}^{(m_1)}\right) = \sum_{\ell=0}^{E-1} \vec{\phi}_{\ell} \eb_\ell^{(E)}\in \mathbb{C}^E,
\end{align}
where $\ell = i_1 + m_1 i_{2} =: \ell_{ind}(i_2, i_1)$. The elements in the essential state vector are defined from the total state vector by $\phi_{i_2, i_1} = \psi_{i_2,i_1}$, for $i_1 \in[0, m_1-1]$ and $i_2 \in [ 0, m_2-1]$.

The initial condition for the solution operator matrix $U(t)$ in Schr\"odinger's equation \eqref{eq:schrodinger_matrix} needs to span a basis for the $E$-dimensional essential state space. Here we use the canonical basis consisting of the unit vectors $\eb_{\ell}^{(E)}$. Let the columns of the initial condition matrix be $U_0=[\gb_0, \gb_1, \ldots, \gb_{E-1}] \in \mathbb{R}^{N\times E}$. Because the total probabilities in each column vector $\gb_k$ must sum to one, the basis vectors in the total state space become
\begin{align}
  \gb_{\ell} = U_0 \eb^{(E)}_{\ell},\quad
    g_{k,\ell} = \begin{cases}
    1,& k = k_{ind}(i_2(\ell), i_1(\ell)),\\
    0,& \mbox{otherwise},
    \end{cases}\quad \mbox{for $\ell = 0,1,\ldots,E-1$}.
\end{align}
Here, $i_2(\ell)= \lfloor \ell/m_1 \rfloor$ and $i_1(\ell) = \ell - m_1\cdot i_2(\ell)$.

The target gate matrix $V_E\in\mathbb{C}^{E\times E}$ defines the unitary transformation between the essential levels in the initial and final states, $\phib_T = V_E \phib_0$, for all $\phib_0\in\mathbb{C}^{E}$. 
Because $V_E$ is unitary, each of its columns has norm one. To preserve total probabilities, we define the target gate transformation according to
\begin{align}\label{eq_target}
  V_{tg} = U_0 V_E \in \mathbb{C}^{N\times E}.
\end{align}
This implies that each column of $V_{tg}$ also has norm one.

\section{The Hamiltonian in a rotating frame of reference}\label{app_RotatingFrame}

The time-dependent and unitary change of variables
$\widetilde{\psib}(t) = R(t)\psib(t)$ where $R^\dag R = I$,
results in the transformed Schr\"odinger equation \begin{equation}\label{eq_timedep_trans}
\frac{d\widetilde{\psib}}{dt} = -i \widetilde{H}(t) \widetilde{\psib},\quad \widetilde{H}(t) = R(t)H(t)R^\dag(t) + i \dot{R}(t) R^\dag(t).
\end{equation}
The rotating frame of reference is introduced by taking the unitary transformation to be the matrix \eqref{eq_rot-trans}. Because both $R(t)$ and the system Hamiltonian \eqref{eq_hamsys} are diagonal, $R H_s R^\dagger = H_s$. The time derivative of the transformation can be written 
\begin{align}
    \dot{R}(t) = \left(\bigoplus_{q=Q}^1 i \omega_{r,q} A_q^\dagger A_q \right) \left(\bigotimes_{q=Q}^1\exp{\left(i\omega_{r,q} t A_q^\dagger A_q \right)}\right), 
\end{align}
where $\oplus$ denotes the Kronecker sum, $C \oplus D = C \otimes I_D + I_C \otimes D$.
Therefore,
\begin{align}
   \dot{R}(t) R^\dagger(t) = \bigoplus_{q=Q}^1 i \omega_{r,q} A_q^\dagger A_q = \sum_{q=1}^Q i\omega_{r,q} a_q^\dagger a_q.
\end{align}

As a result, the first term in the Hamiltonian \eqref{eq_hamsys} is modified by the term $i \dot{R}(t) R^\dag (t)$. After noting that $R a_q = e^{-i\omega_{r,q} t} a_q R$, the transformed Hamiltonian can be written as
\begin{align}
  H^{rw}_s &= \sum_{q=1}^Q  \left(\Delta_q a_q^\dagger a_q -\frac{\xi_q}{2} a_q^\dag a_q^\dag a_q a_q - \sum_{p>q} \xi_{qp} a_q^\dag a_q a_p^\dag a_p  \right), \\
  \widetilde{H}_c(t) &= \sum_{q=1}^Q f_q(t;\alphab) \left( e^{-i\omega_{r,q} t} a_q + e^{i\omega_{r,q} t} a_q^\dag \right),\label{eq_trans-hamiltonian}
\end{align}
where $\Delta_q = \omega_q - \omega_{r,q}$ is the detuning frequency. The above system Hamiltonian corresponds to \eqref{eq_hamsysrot}.

To slow down the time scales in the control Hamiltonian, we want to absorb the highly oscillatory factors $\exp(\pm i\omega_{r,q} t)$ into $f_q(t)$. Because the control function $f_q(t)$ is real-valued, this can only be done in an approximate fashion. We make the ansatz,
\begin{align}
  f_q(t) := 2\, \mbox{Re}\left( d_q(t) e^{i\omega_{r,q} t} \right)
  = d_q(t)e^{i\omega_{r,q} t} + \bar{d}_q(t)e^{-i\omega_{r,q} t},
\end{align}
where $\bar{d}_q$ denotes the complex conjugate of $d_q$. By substituting this expression into the transformed control Hamiltonian \eqref{eq_trans-hamiltonian}, we get
\begin{align*}
  \widetilde{H}_c(t) 
  = \sum_{q=1}^Q \left( d_q(t) a_q + \bar{d}_q(t) a_q^\dagger 
  +  \bar{d}_q(t) e^{-2i\omega_{r,q} t} a_q  + d_q(t)e^{2i\omega_{r,q} t} a_q^\dag \right).
\end{align*}
The rotating wave approximation (RWA) follows by ignoring terms that oscillate with
frequency, $\pm 2i\omega_{r,q}$, resulting in the approximate control Hamiltonian \eqref{eq_hamctrlrot}.


\section{Conditions for resonance}\label{app_resonance}

Consider the scalar function $y(t) := \psi^{(1)}_{\jb}(t)$. It satisfies an ordinary differential equation of the form
\begin{align}\label{eq_scalar-first-order}
    \frac{d y(t)}{dt} +\kappa_{\jb} y(t) = \sum_{\ell} c_\ell e^{i\nu_\ell t},\quad \nu_k\in\mathbb{R}.
\end{align}
We are interested in cases when $y(t)$ grows in time, corresponding to resonance. Conditions for resonance are provided in the following lemma.
\begin{lemma}\label{lem_resonance}
Let $\kappa\in\mathbb{R}$ and $\nu\in\mathbb{R}$ be constants. The solution of the scalar ordinary differential equation
\begin{align}
    \frac{d y(t)}{d t} + i\kappa y(t) = c e^{i\nu t},\quad y(0) = y_0,
\end{align}
is given by 
\begin{align}
y(t) = \begin{cases}
    y_0 e^{-i\kappa t} + c t e^{-i\kappa t},\quad & \nu+\kappa = 0,\\
    y_0 e^{-i\kappa t} - \dfrac{ic}{\nu + \kappa} \left( e^{i\nu t} - e^{-i\kappa t}\right),&\mbox{otherwise}.
    \end{cases}
\end{align}
Corresponding to resonance, the function $y(t)$ grows linearly in time when $\nu+\kappa=0$ and $c\ne 0$.
\end{lemma}
\begin{proof}
Follows by direct evaluation.
\end{proof}

We proceed by analyzing the right hand side of \eqref{eq_psi-one}.
It can be shown that the forcing function $\fb^{(k)}(t)$ is of the form
\begin{align}
    \fb^{(k)}_{\jb}(t) &= \begin{cases}
   -i g_{\jb+\eb_k}\sqrt{j_k+1}\, e^{i(\Omega_k-\kappa_{\jb+\eb_k})t},\quad & j_k=0,\\
   \Theta_k(t),\quad &j_k\in[1,n_k-2],\\
   -i g_{\jb-\eb_k}\sqrt{j_k}\, e^{-i(\Omega_k +\kappa_{\jb-\eb_k})t},\quad & j_k = n_k-1,
  \end{cases}
\end{align}
where $\Theta_k(t) = -i g_{\jb+\eb_k}\sqrt{j_k+1}\, e^{i(\Omega_k-\kappa_{\jb+\eb_k})t} 
  - i g_{\jb-\eb_k}\sqrt{j_k}\, e^{-i(\Omega_k +\kappa_{\jb-\eb_k})t}$.

The right hand side satisfies $\fb(t) = \fb^{(1)}(t) + \fb^{(2)}(t)$.
The first set of frequencies and coefficients on the right hand side of \eqref{eq_scalar-first-order} satisfy
\begin{align*}
    \nu_{1} = 
    \Omega_k-\kappa_{\jb+\eb_k},\quad
    c_{1} = -i g_{\jb+\eb_k}\sqrt{j_k+1},
\end{align*}
for $k=\{1,2\}$ and $j_k\in[0,n_k-2]$. The second set of frequencies and coefficients are 
\begin{align*}
    \nu_{2} = 
    -(\Omega_k+\kappa_{\jb-\eb_k}),\quad
    c_{2} = -i g_{\jb-\eb_k}\sqrt{j_k}.
\end{align*}
for $k=\{1,2\}$ and $j_k \in[1,n_k-1]$.
From Lemma~\ref{lem_resonance}, component 
$\psi^{(1)}_{\jb}(t)$ is in resonance if 
$(\kappa_{\jb}+\nu_{1}=0, c_{1}\ne 0)$ or 
$(\kappa_{\jb}+\nu_{2}=0, c_{2}\ne 0)$. 
These conditions are equivalent to \eqref{eq_cond1} and \eqref{eq_cond2}, which proves Lemma~\ref{lem_freq}.


\section{Non-diagonal system Hamiltonians}\label{app_JC-Ham}

Current superconducting quantum computing hardware  consists of qubits coupled by high quality resonators. Here we consider a transmon (artificial atom) qubit in which a Josephson junction is coupled in parallel to a capacitor. This section provides a brief overview of the derivation of the Hamiltonian model in the dispersive limit. For further details, see for example the review paper by Blais et al.~\cite{BlaisEtal-21}. The interaction between the transmon qubit and the resonator can be modeled by the Hamiltonian~\cite{BlaisEtal-21},
\begin{align}
 H_{qr} = \omega_r a^\dagger a + \omega_q b^\dagger b - \frac{E_C}{2} b^\dagger b^\dagger b b + g(b^\dagger - b)(a^\dagger - a),   
\end{align}
where $\omega_r$ is the resonant frequency of the resonator, $\omega_q$ is the transmon frequency, $E_C$ is the charging energy, and $g$ is the resonator-transmon coupling constant. Further, $a$ and $b$ are the respective lowering operators of the resonator and the transmon. The Hamiltonian can be simplified when the coupling constant is small compared to the resonant and transmon frequencies ($|g|\ll \omega_r$, $|g|\ll \omega_q$), resulting in
\begin{align}
 H_{qr} \approx H_{lin} + H_{nl},\quad H_{lin} =\omega_r a^\dagger a + \omega_q b^\dagger b - \frac{E_C}{2} b^\dagger b^\dagger b b,\quad H_{nl} = g(b^\dagger a + b a^\dagger).   
\end{align}
Let $\Delta = \omega_q - \omega_r$ be the qubit-resonator detuning frequency and define $\lambda = g/\Delta$. In the dispersive limit, when $|\lambda| = |g/\Delta| \ll 1$, the qubit and resonator are only weakly entangled and the Hamiltonian can be diagonalized by a Schrieffer-Wolff~\cite{SW-Bravyi-11, SW-original} expansion, or by using the Bogoliubov transformation~\cite{Boisson-Etal-09}. In the latter case, we take the unitary transformation to be $U_{disp} = \exp(\Lambda(a^\dagger b - a b^\dagger))$. Under this transformation, the lowering operators transform as $U_{disp}^\dagger a U_{disp} = \cos(\Lambda) a + \sin(\Lambda) b$ and $U_{disp}^\dagger b U_{disp} = \cos(\Lambda) b - \sin(\Lambda) a$. By choosing $\tan(2\Lambda) = 2\lambda$, the linear part of the Hamiltonian transforms to diagonal form,
\begin{align}
    U_{disp}^\dagger H_{lin} U_{disp} = \tilde{\omega}_r a^\dagger a + \tilde{\omega}_q b^\dagger b,
\end{align}
where the so-called ``dressed" frequencies are $\tilde{\omega}_r = 0.5(\omega_r + \omega_q - \sqrt{\Delta^2 + 4g^2})$ and $\tilde{\omega}_q = 0.5(\omega_r + \omega_q + \sqrt{\Delta^2 + 4g^2})$. After an additional Schrieffer-Wolff expansion in $|\lambda| \ll 1$, the non-linear part of the Hamiltonian can also be diagonalized, resulting in the dispersive Hamiltonian
\begin{multline}
    U_{disp}^\dagger (H_{lin} + H_{nl}) U_{disp}\approx \\
     \tilde{\omega}_r a^\dagger a + \tilde{\omega}_q b^\dagger b +  \frac{K_a}{2} a^\dagger a^\dagger a a +  \frac{K_b}{2} b^\dagger b^\dagger b b + \chi_{ab} a^\dagger a b^\dagger b =:  H_{disp}.
\end{multline}
In the dispersive regime, the coefficients $K_a$, $K_b$ and $\chi_{ab}$ are all negative. The dispersive Hamiltonian is diagonal. This approach can be generalized to several transmons and resonators~\cite{BlaisEtal-21}, resulting in the model~\eqref{eq_hamsys}.

\bibliographystyle{siam}
\bibliography{quantum}
\end{document}